\documentclass[11pt,reqno]{article}

\usepackage{amsmath,amssymb,amsfonts,amsthm,bbm,mathrsfs,verbatim} 
\usepackage{array}
\usepackage{setspace}
\usepackage{authblk}
\usepackage{cancel}
\usepackage{caption}
\usepackage{geometry}
\usepackage{graphics} 
\usepackage{graphicx}  
\usepackage{multirow}
\usepackage[utf8x]{inputenc}
\usepackage{enumerate}
\usepackage{cancel}
\usepackage{hyperref}

\usepackage[authoryear]{natbib}
\usepackage{lipsum}
\usepackage{color}
\usepackage{bm}

\newtheorem{theorem}{Theorem}
\newtheorem{proposition}{Proposition}
\newtheorem{assumption}{Assumption}

\theoremstyle{definition}
\newtheorem{remark}{Remark}
\newtheorem{definition}{Definition}
\newtheorem{example}{Example}

\allowdisplaybreaks
\def\R{{\mathbb R}}

\def\E{{\mathbb E}}

\newcommand{\nnb}{\nonumber}

\definecolor{dblue}{RGB}{12, 56, 100}

\title{\bf A tail-shape actuarial index based on equal level relationships between Value at Risk and Expected Shortfall}

\author[a,b]{Georgios I. Papayiannis\footnote{Corresponding author e-mail: \url{gpapagiannis@unipi.gr} } }
\author[a]{Georgios Psarrakos\footnote{\url{gpsarr@unipi.gr}}}
\affil[a]{\small Department of Statistics \& Insurance Science, University of Piraeus, Piraeus, GR}
\affil[b]{\small Stochastic Modelling and Applications Laboratory, Athens University of Economics \& Business, Athens, GR}

\begin{document}
	\graphicspath{ {Figures/} } 
	\maketitle

\begin{abstract}		
	
	We introduce a new actuarial tail-shape index, the $\theta$-index, based on a probability equal level relationship between Value at Risk and Expected Shortfall. The index is defined at each tail probability level as the parameter value for which Value at Risk coincides with Flexible Expected Shortfall, that is a convex mixture of Expected Shortfall and the mean. This yields a level-dependent, scale-free measure of upper tail behaviour. We study basic theoretical properties of the $\theta$-index and introduce a partial order for comparing loss distributions, characterized by the monotonicity of right-tail spread ratios. Additionally, the index leads to characterizations of the tail behaviour of a loss distribution as consistent to the generalized Pareto model, through a direct connection to the mean excess function. Moreover, we derive Euler risk contributions for the $\theta$-index and use probability equal level relationships to compute Value at Risk allocations in a more stable way. Finally, the $\theta$-index is examined as a diagnostic tool for distinguishing tail regimes and its capabilities are illustrated using the Danish fire insurance dataset.
\end{abstract}
	

%\tableofcontents

%%
%%
\section{Introduction}\label{Sec-1}

Quantifying tail risk is a central task in actuarial science and risk management. Measures such as Value at Risk (VaR) and Expected Shortfall (ES) are widely used to assess the severity of extreme losses and play an important role in regulation and practice. While these measures provide information at fixed probability levels, they offer limited insights into how tail behaviour evolves as the conditioning level changes. Understanding this local evolution of tail severity is important in standard actuarial operations such as reinsurance design, retention selection, and risk capital allocation.

A broad literature addresses the assessment of tail risk and tail behaviour in actuarial science and risk management. Standard risk measures such as VaR and ES summarize tail severity at fixed probability levels \citep{embrechts2013modelling, mcneil2015quantitative}. Distortion and spectral risk measures provide flexible aggregation mechanisms grounded in stochastic orders and weighted premium principles \citep{denneberg1994non, wang1996premium, wang1998actuarial}. From an asymptotic perspective, extreme value theory and regular variation techniques yield well-established characterizations of tail heaviness at very high probability levels \citep{de2006extreme, resnick2007heavy}. In applied actuarial analysis, diagnostics based on the mean excess function and related tail index estimators are widely used to explore tail behaviour across thresholds \citep{albrecher2017reinsurance,beirlant2006statistics}. Closely related level-dependent tail indices include the Expected Proportional Shortfall \citep{belzunce2012comparison}, which is based on ratios of tail excess quantities. More recently, probability equivalent level relationships between VaR and ES have been studied by characterizing the probability levels at which these measures coincide \citep{li2023pelve}. 

Building on these developments, this paper adopts a complementary perspective on tail behaviour. While level-dependent indices and probability equivalent level relationships provide useful tools for tail comparison, many classical tail diagnostics are motivated by asymptotic arguments and are primarily informative at extreme probability levels. Their interpretation at moderate probability levels, which are often relevant for actuarial decisions, can therefore be delicate. In this paper, we focus on a non-asymptotic, level-based description of tail behaviour through a representation that links VaR to a flexible mixture of ES and the mean at the same probability level. This leads to the $\theta$-index, a level-dependent tail-shape measure that summarizes how conditional tail severity evolves beyond a given quantile. The index is designed to be interpretable, scale-free, and directly connected to ordering, allocation, and tail diagnostics.

The remainder of the paper is organized as follows. Section \ref{Sec-2} introduces the main concepts and presents the definition of the $\theta$-index. In Section \ref{Sec-3}, basic theoretical properties of the $\theta$-index are studied and a partial order is introduced for comparing loss distributions. In Section \ref{Sec-4}, implications for Euler risk allocations are studied, while in Section \ref{Sec-5} a tail diagnostic framework is developed for distinguishing tail regimes based on $\theta$-index related measures. Finally, in Section \ref{Sec-6} the proposed methodology for common loss models is illustrated while in Section \ref{Sec-7} an empirical tail analysis to the Danish fire insurance dataset is performed employing the $\theta$-index framework. 

\section{Main concepts and the $\theta$-index}\label{Sec-2}

Weighted premiums and mixing representations are widely used in actuarial tail modelling and provide flexible mechanisms for combining tail-conditional quantities. Recently,  \cite{psarrakos2024flexibility} introduced the class of flexible weighted premiums of the form
\begin{equation}\label{weighted-premia}
	\Pi_w(X;\theta) = \frac{ \E[X \, w_{\theta}(X) ] }{ \E[w_{\theta}(X)] },
	%= \frac{ \E[X \, w(X)] + \theta \, \E[X] }{ \E[w(X)] + \theta },
\end{equation}  
where the employed weighting function $w_{\theta}(x) = w(x) + \theta$, depends on a non-decreasing weighting function $w(x)$ and an extra parameter $\theta \in (0,\infty)$ referred to as the flexibility parameter. For the indicator weighting function $w(x) = {\bf 1}\{ x > \mbox{VaR}_p(X) \}(x)$ \citep{furman2008weighted}, this construction yields a mixture of the tail conditional expectation and the mean, with mixing proportions determined jointly by the probability level $p$ and the parameter $\theta$ expressed by 
\begin{equation}\label{mixture}
	\frac{1-p}{1-p+\theta} \, \mbox{ES}_p(X) + \frac{\theta}{1-p+\theta} \, \mathbb{E}[X]. 
\end{equation}	
This mixing perspective also underlies the class of coherent risk measures known as ES/$\mathbb{E}$ - mixtures \citep{embrechts2021bayes}, as well as the recently proposed Flexible Conditional Tail Expectation (FCTE) or Flexible Expected Shortfall (FES), where the relative weight assigned to ES and the mean depends on the pair $(p,\theta)$. These frameworks illustrate that mixture representations linking ES, VaR and the mean, encode structural information about the behaviour of the upper part of the distribution. In particular, the FES functional leads to a natural level-dependent relationship between ES and VaR that can be characterised explicitly.

For what follows we consider random variables with finite and non-negative mean with continuous and strictly increasing distribution function on its support. In particular, on a probability space $(\Omega, \mathcal{F}, \mathbb{P})$ we focus on the class of random variables
\begin{equation*}
	\mathcal{X} := \{ X:\Omega \to \R \,\, : \,\, \E[|X|]<\infty, \,\,\, \E[X] \geq0, \,\,\, F_X \mbox{ continuous \& strictly increasing } \} \subset L^1
\end{equation*}
with $F_X(\cdot)$ denoting the distribution functions of $X$.  Moreover, we define the upper tail level set of $X$ as
\begin{equation}
	D_X := \{ p \in (0,1) \,\, : \,\, \mbox{VaR}_p(X) > \E[X] \geq 0 \}.
\end{equation} 
Throughout the paper, results that involve differentiation with respect to the probability level $p$ are derived under the following additional assumption.
\begin{assumption}[Regularity assumption]\label{ass-1}
	 The distribution function $F_X$ admits a density $f_X$ which is positive at the relevant quantile levels, i.e.
$$ f_X(\mathrm{VaR}_p(X))>0, \qquad p \in D_X.$$
\end{assumption}
We recall that under the aforementioned conditions for $F_X$, the VaR coincides with the quantile function of $X$, while the ES and CTE risk measures are identical.

\begin{proposition}\label{prop-1}
	For a random variable $X \in \mathcal{X}$ and any fixed $\theta \in(0,\infty)$, there exists a unique $p_{\theta} \in (0,1)$ such that
	$$ \mbox{FES}_{p_{\theta}}(X;\theta) = 	\frac{1-p_{\theta}}{1-p_{\theta}+\theta} \, \mbox{ES}_{p_{\theta}}(X) + \frac{\theta}{1-p_{\theta}+\theta} \, \mathbb{E}[X] = \mbox{VaR}_{p_{\theta}}(X)$$
	where $p_{\theta} := \arg\max_{p \in(0,1)} \mbox{FES}_p(X)$.
\end{proposition}

\begin{proof}	
	The result is an immediate consequence of Lemma 3.5 in \cite{psarrakos2024flexibility}.
	%By derivation of FES with respect to the level $p$ we get
	%\begin{eqnarray*}
	%	\frac{d}{dp}\mbox{FES}_p(X;\theta) = \frac{1}{1-p+\theta}\left[ \mbox{FES}_p(X;\theta) - \mbox{VaR}_p(X) \right]
	%\end{eqnarray*}
	%i.e. any level $p \in (0,1)$ for which the equality $\mbox{FES}_{p}(X;\theta) = \mbox{VaR}_{p}(X)$ holds, is a critical point of FES. However, since for any $p \in (0,1)$ it holds that 
	%$$ \frac{d^2}{dp^2}\mbox{FES}_p(X;\theta) = -\frac{1}{1-p+\theta}\left[ \mbox{VaR}_{p}(X) \left( 1 + \frac{1}{1-p+\theta}\right) + \frac{1}{f(\mbox{VaR}_{p}(X))} \right] < 0,$$   
	%the strict concavity of FES indicates that the critical point $p_{\theta}$ of FES, is its unique maximizer. 
\end{proof}

We stress that the equality established above concerns the direct relation of VaR and FES rather than VaR and ES. The unique maximizer $p_{\theta}$ identifies the probability level at which the FES functional coincides with the VaR. This yields a level-dependent alignment between FES and VaR, governed by the parameter $\theta$. Building on the mixture representation underlying FES and on the VaR-FES alignment at the maximizing level $p_{\theta}$, we obtain a level-dependent relationship between ES and VaR that can be inverted to define a tail-shape index. 

\color{black}

%Throughout the rest of the paper we consider random variables $X$ satisfying:
%\begin{itemize}
%	\item[(i)] $0 \leq \E[X] < \infty$,
%	\item[(ii)] the distribution of $X$ is continuous
%\end{itemize}
%and we define the upper tail level set of $X$ as
%\begin{equation}
%	D_X := \{ p \in (0,1) \,\, : \,\, \mbox{VaR}_p(X) > \E[X] \geq 0 \}.
%\end{equation}
%Under the continuity condition, the VaR coincides with the quantile function of $X$, while the ES and CTE risk measures are identical.

\begin{definition}[$\theta$-index]\label{theta-index}
	Assume a random variable $X \in \mathcal{X}$. For any level $p \in D_X$, its $\theta$-index is defined by the function
	\begin{equation}\label{theta}
		\theta_p(X) := (1-p) \frac{ \mbox{ES}_p(X) - \mbox{VaR}_p(X) }{ \mbox{VaR}_p(X) - \E[X] }  
		%\mathbb{E}\left[ \left( \frac{X - \mbox{VaR}_p(X)}{ VaR_p(X - \mathbb{E}[X]) }\right)_+ \right] 
		%\frac{  \pi_{X}(\mbox{VaR}_p(X)) }{ \mbox{VaR}_p(X) - \mathbb{E}[X]} 
		= (1-p) \frac{ e_X(\mbox{VaR}_p(X)) }{ \mbox{VaR}_p(X) - \E[X] }
	\end{equation}  
	where %$\pi_{X}(z) := \mathbb{E}[(X-z)_+]$ denotes the stop loss transform for $z \geq 0$, and 
	$e_X(z) = \E[X-z|X>z]$ denotes the mean excess function of $X$. 
	\begin{comment}
		Equivalently, for any threshold $u > \E[X]$, the $\theta$-index of $X$ is also represented by 
		\begin{equation}
			\theta_{X}(u) := \frac{ \pi_X(u) }{ u - \E[X] }  = \bar{F}_X(u) \frac{ e_X(u) }{ u - \E[X] }. 
		\end{equation}
		The equivalence of the two definitions is obtained for $u = \mbox{VaR}_p(X)$ for any $p \in D_X$.
	\end{comment}
\end{definition}

Conceptually, the $\theta$-index identifies at each probability level, the parameter value for which the Flexible Expected Shortfall coincides with VaR, as described in Proposition \ref{prop-1}. Definition \ref{theta-index} provides a more direct interpretation of the index in terms of tail behaviour. From equation \eqref{theta}, $\theta_p(X)$ measures the excess of the tail conditional expectation above $\mbox{VaR}_p(X)$, scaled by the distance between $\mbox{VaR}_p(X)$ and the mean. As a result, $\theta_p(X)$ quantifies how tail severity increases as the conditioning level moves further into the tail. In the next sections, we study basic analytical properties of the $\theta$-index and show how monotonicity, curvature, and transformation properties provide additional insight into the structure of the upper tail.
		
\color{black}	

%%%%%%%%%%%%%%%%%%%%%%%%%%%%%%%%%%%%%%%%%%%%%%%%%%%%%%%%%%%%%%%%%%%%%%%%%%%%%%%%%%%%%%%%%%%%
\section{Theoretical properties of the $\theta$-index and characteri\-zation results}\label{Sec-3}
%%%%%%%%%%%%%%%%%%%%%%%%%%%%%%%%%%%%%%%%%%%%%%%%%%%%%%%%%%%%%%%%%%%%%%%%%%%%%%%%%%%%%%%%%%%%

In this section we develop theoretical properties of the $\theta$-index. We first establish its invariance under affine transformations, which allows comparisons of tail behaviour across risks with different units, scales, or contractual shifts. We then study the effect of convex and concave transformations, derive monotonicity properties, and introduce a nonparametric estimator together with its convergence properties. Building on these results, we define a partial order induced by the $\theta$-index accompanied by a comparison theorem across different models. The section concludes with a characterisation result of the upper tail behaviour through the generalized Pareto distribution.

We begin with the affine invariance property of $\theta_p(X)$ which is stated in the following result.
	
\begin{proposition}\label{prop-ls-inv}
	Assume $X \in \mathcal{X}$. Given any $p \in D_X$, $\theta$-index satisfies the property of affine invariance, i.e.
	$$ \theta_p(\alpha \, X + \beta) = \theta_p(X) $$
	for any $\alpha >0 $ and $\beta \in \R$.  
\end{proposition}	
	
\begin{proof}
	Employing the affine invariance properties of the mean, VaR and ES we obtain the following:
	\begin{eqnarray*} 
		\theta_p(\alpha \, X + \beta)  
		&=& (1-p)  \frac{\mbox{ES}_p(\alpha \, X + \beta) - \mbox{VaR}_p(\alpha \, X + \beta)}{\mbox{VaR}_p(\alpha \, X + \beta) - \mathbb{E}[\alpha \, X + \beta]}  \\
		&=&  (1-p) \frac{\alpha \,  \mbox{ES}_p(X) + \beta  - (\alpha \, \mbox{VaR}_p(X) + \beta)}{\alpha \, \mbox{VaR}_p(X) + \beta  - (\alpha \, \mathbb{E}[X] + \beta)}  
		=  \theta_p(X)   
	\end{eqnarray*}
\end{proof}

Our next result, concerns the effects of convex and concave transforms of the underlying random variable on $\theta$-index. 

\begin{proposition}\label{prop-3}
	Assume $X \in \mathcal{X}$. For any $p \in D_X$, $\theta$-index satisfies the following properties:
	\begin{itemize}
		\item[(a)] If the function $g:\R \to \R$ is convex and increasing, it holds that $\theta_p( g(X) ) \geq \theta_p(X)$.
		\item[(b)] If the function $g:\R \to \R$ is concave and increasing, it holds that $\theta_p(g(X)) \leq \theta_p(X)$.
	\end{itemize}
\end{proposition}

\begin{proof}
	
	(a). Since $g$ is an increasing and convex function, it has an increasing derivative, i.e. for any $y<x$ it holds that $g'(y) \leq g'(x)$. Combining this with the Mean Value Theorem (MVT) we have that 
	\begin{eqnarray*}
		g(x) - g(y) &=& g'(\xi) \, (x-y), \,\,\, \xi \in [y,x]  \\
		 &\geq& g'(y) (x-y), \,\,\, x \geq y
	\end{eqnarray*} 
	since necessarily $\xi \geq y$.
	%For any $y < x$, applying the Mean Value Theorem (MVT) and the property that an increasing and convex function has a non-decreasing derivative, we get the inequality
	%$$ g(x) - g(y) \geq g'(y) (x-y) $$
	%for any $x \geq y$. 
	Taking conditional expectation to both sides and setting $y = VaR_p(X)$ for any $p \in D_X$, we get
	\begin{eqnarray*}
		e_{g(X)}(g(\mbox{VaR}_p(X))) &:=& \E[g(X) - g(\mbox{VaR}_p(X)) | g(X) > g(\mbox{VaR}_p(X))]\\ 
		&\geq& g'( \mbox{VaR}_p(X) ) \E[X - \mbox{VaR}_p(X) | X > \mbox{VaR}_p(X)] \\
		&=:& g'( \mbox{VaR}_p(X) ) \, e_X(\mbox{VaR}_p(X)).
	\end{eqnarray*}
	Employing Jensen's inequality, the fact that $g$ is convex and the property $g(\mbox{VaR}_p(X)) = \mbox{VaR}_p(g(X))$, we get the inequality
	\begin{equation*}
		g(\mbox{VaR}_p(X)) - g(\E[X]) \geq g(\mbox{VaR}_p(X)) - \E[g(X)].
	\end{equation*}
	Then, using the MVT we have
	\begin{eqnarray*}
		g(\mbox{VaR}_p(X)) - g(\E[X]) &=& g'(\xi) ( \mbox{VaR}_p(X) - \E[X]), \,\,\, \xi \in [ \E[X], \mbox{VaR}_p(X)]  \\
		&\leq& g'(\mbox{VaR}_p(X)) \, ( \mbox{VaR}_p(X) - \E[X] ), \,\,\, \mbox{VaR}_p(X) > \E[X]
	\end{eqnarray*}
	since necessarily $\xi \leq \mbox{VaR}_p(X)$. Combining with the previous result and the property $g(\mbox{VaR}_p(X)) = \mbox{VaR}_p(g(X))$ since $g$ increasing, we get that 
	\begin{eqnarray*}
		  \mbox{VaR}_p(g(X)) - \E[g(X)] &\leq& g'(\mbox{VaR}_p(X)) \, ( \mbox{VaR}_p(X) - \E[X] ) \iff \\
		  \frac{1}{ \mbox{VaR}_p(g(X)) - \E[g(X)] } &\geq& \frac{1}{ g'(\mbox{VaR}_p(X))  } \, \frac{1}{ \mbox{VaR}_p(X) - \E[X] },
	\end{eqnarray*}	
	and multiplying with the obtained inequality for mean excess loss functions, we reach the relation $\theta_p( g(X) ) \geq \theta_p(X)$ which concludes the proof. \\ \\
	(b) Taking into account the property that if $g$ is increasing and concave, it has decreasing derivative, i.e. $g'(y) \geq g'(x)$ for any $y < x$, and following the same steps as in (a) we reach the stated result.
\end{proof}

Proposition \ref{prop-3} formalises how the $\theta$-index responds to increasing transformations of losses that alter the relative weight of extreme outcomes. In particular, convex increasing transformations, such as non-linear indexation clauses or convex catastrophe loadings, place greater emphasis on large claims and lead to larger values of $\theta_p$ comparing to the original loss. By contrast, concave increasing transformations associated with contractual features, such as policy limits or caps, compress large losses relative to moderate ones and result in smaller values of $\theta_p$. These properties align with standard actuarial interpretations of how convexity and concavity affect tail-sensitive risk measures such as VaR and ES.

Next, we provide the non-negativity and monotonicity (decreasing) property of $\theta$-index for $p \in D_X$. 	 
\begin{proposition}\label{theta-decreasing}
	Assume $X \in \mathcal{X}$. The $\theta$-index is positive and strictly decreasing for any $p \in D_X$.
\end{proposition}
\begin{proof}
	It is easy to verify that $\theta_p(X) \in (0,\infty)$ by definition for any $p \in D_X$. Moreover, for any $p \in D_X$ we get that 
	\begin{eqnarray*}
		\frac{d}{dp}\theta_p(X) &=& \frac{d}{dp} \left[ (1-p) \frac{\mbox{ES}_p(X) - \mbox{VaR}_p(X)}{ \mbox{VaR}_p(X) - \mathbb{E}[X] } \right] \\
		&=& -\frac{\mbox{ES}_p(X) - \mbox{VaR}_p(X)}{\mbox{VaR}_p(X) - \mathbb{E}[X]} + (1-p)\left( \frac{\frac{d}{dp}[\mbox{ES}_p(X) - \mbox{VaR}_p(X)]}{ \mbox{VaR}_p(X) - \mathbb{E}[X] } \right. \\
		&\,\,\,& \left.- \frac{(\mbox{ES}_p(X) - \mbox{VaR}_p(X)) \frac{d}{dp}[\mbox{VaR}_p(X)
			 - \mathbb{E}[X]] }{ (\mbox{VaR}_p(X) - \mathbb{E}[X])^2 } \right) \\
		 &=& -\frac{\theta_p(X)}{1-p} + \frac{\theta_p(X)}{1-p}  - \frac{1-p}{f_X(\mbox{VaR}_p(X))} \left( 1 +  \frac{ \mbox{ES}_p(X) - \mbox{VaR}_p(X) }{ \mbox{VaR}_p(X) - \mathbb{E}[X] } \right) \\
		 &=& -\frac{1 - p + \theta_p(X)}{ \mbox{VaR}_p(X) - \mathbb{E}[X] } \frac{1}{f_X(\mbox{VaR}_p(X))} < 0
	\end{eqnarray*}
	since $\theta_p(X) \in (0,\infty)$, $\mbox{VaR}_p(X) > \mathbb{E}[X]$ and $f_X(\mathrm{VaR}_p(X)) > 0$ for any $p \in D_X$. Therefore, the strict negativity of $\frac{d}{dp}\theta_p(X)$ indicates that $\theta_p(X)$ is strictly decreasing for any $p \in D_X$. 
\end{proof}

Wherever there is not adopted any particular parametric model for the description of $X$, an empirical estimator for $\theta_p(X)$ may be employed.  Let us consider the empirical version of $\theta_p(X)$ which is denoted by $\hat{\theta}_{p,n}(X)$ to emphasize the dependence on the sample size $n$. Similarly to \cite{belzunce2012comparison}, the empirical estimator for a given iid sample $X_1,X_2,...,X_n$ from the distribution of $X$ is given by
\begin{equation}\label{theta-emp}
	\widehat{\theta}_{p,n}(X) = \frac{(1-p)}{ \hat{x}_{p,n} - \bar{X} } \sum_{i=1}^n \frac{ (X_i - \hat{x}_{p,n}) I_{( \hat{x}_{p,n} ,\infty)}(X_i) }{ \sum_{j=1}^n I_{(\hat{x}_{p,n},\infty)}(X_j) }.
\end{equation}
Here, $\hat{x}_{p,n} := X_{([n(1-p)]+1)}$ denotes the empirical estimator of the quantile function at level $p$ (equivalently, the empirical $\mathrm{VaR}_p(X)$), and $\bar{X}$ denotes the sample mean. Throughout, we consider the natural plug-in estimator obtained by replacing population quantities with their empirical counterparts. In the next proposition we establish strong consistency of the estimator $\hat{\theta}_{p,n}(X)$.

\begin{theorem}\label{prop-convergence}
	Assume $X \in \mathcal{X}$. Let $\widehat{\theta}_{p,n}(X)$ be the empirical estimator defined in \eqref{theta-emp}. For any $p \in D_X$, it holds that $$\widehat{\theta}_{p,n}(X) \rightarrow \theta_p(X) \quad \mbox{a.s.}$$
\end{theorem}

\begin{proof}
	%Let us denote for simplicity $es_p := \mbox{ES}_p(X)$, $x_p := \mbox{VaR}_p(X)$, $\mu := \mathbb{E}[X]$ and by $\widehat{es}_{p,n}, \hat{x}_{p,n}, \hat{m}_n$ the respective empirical estimators (where $n$ highlights the dependence to the sample size $n$). 
	
	%First we state some convergence results concerning the empirical quantities $\widehat{es}_{p,n}, \hat{x}_{p,n}, \hat{m}_n$. The empirical mean statistic, it is well known by the Law of Large Numbers (LLN) that $m_n \to m$ almost surely and therefore also in probability. For the empirical quantile estimator $\hat{x}_{p,n} := \hat{F}^{-1}_{n}(p)$, the convergence to the quantile function $F^{-1}(p) =: x_p$ for any $p \in (0,1)$ holds almost surely \citep{gilat1992one, van2000asymptotic} and therefore in probability. Moreover, the empirical estimator of the ES (or CTE) $\widehat{es}_{p,n}$ converges to $es_p$ in probability \citep{brazauskas2008estimating}. Keeping in mind the aforementioned results we have that for any $p \in D_X$:

	Let us denote for notation simplicity $es_p := \mbox{ES}_p(X)$, $x_p := \mbox{VaR}_p(X)$, $\mu := \mathbb{E}[X]$ and by $\widehat{es}_{p,n}, \hat{x}_{p,n}, \hat{\mu}_n$, the respective empirical estimators (referring to an iid sample of size $n$). Working with the absolute difference between the empirical estimator and $\theta_p(X)$ we get
	\begin{eqnarray*}
		|\hat{\theta}_{p,n}(X) - \theta_p(X)| &=& (1-p) \left\lvert \frac{\widehat{es}_{p,n} - \hat{x}_{p,n}}{ \hat{x}_{p,n} - \hat{m}_n } - \frac{es_{p} - x_{p}}{ x_{p} - m } \right\rvert \\
		&=& C_{p,n} \, \left\lvert (\widehat{es}_{p,n} - \hat{x}_{p,n})(x_p - \mu) - (es_p - x_p) (\hat{x}_{p,n} - \hat{\mu}_n) \right\rvert \\
		%& = & \left \lvert \frac{1-p}{ (x_p - m)( \hat{x}_{p,n} - \hat{m}_n) } \right\rvert \left\lvert (\widehat{es}_{p,n} - \hat{x}_{p,n})(x_p - m)  \right. \\
		%&\,\,\,& \left. -  (es_p - x_p) (\hat{x}_{p,n} - \hat{m}_n) \right\rvert \\
		%& = & C_{p,n} | \widehat{es}_{p,n} x_p - \widehat{es}_{p,n} m + \hat{x}_{p,n} m - x_p \hat{m}_n + es_p \hat{m}_n - es_p \hat{x}_{p,n} | \\
		& \leq & C_{p,n} \, \left( |\widehat{es}_{p,n} \,x_p - es_p \,\hat{x}_{p,n} | + |\hat{x}_{p,n} \, \mu - x_p \, \hat{\mu}_n | + |es_{p} \, \hat{\mu}_n - \widehat{es}_{p,n} \, \mu | \right)\\
		%&=& C_{p,n} \left\{  |\widehat{es}_{p,n} x_p + es_p x_p - es_p x_p - es_p \hat{x}_{p,n}|  \right. \\
		%&\,\,\,& \,\,\,\,\,\,\,\,\,\,\, + |\hat{x}_{p,n} m - x_p m + x_p m - x_p \hat{m}_n | \\
		%&\,\,\,& \,\,\,\,\,\,\,\,\,\,\, \left. + | es_{p} \hat{m}_n - es_p m + es_p m - \widehat{es}_{p,n} m | \right\} \\ 
		&\leq& C_{p,n} \left\{  | \widehat{es}_{p,n} - es_p| |x_p| + |\widehat{es}_{p,n} - es_p| |\mu| \right.
		+ |\hat{x}_{p,n} - x_p| | es_p | + |\hat{x}_{p,n} - x_p| | \mu | \\
		&\,\,\,& \,\,\,\,\,\,\,\,\,\,\, + \left. |\hat{\mu}_n - \mu| | es_p | + |\hat{\mu}_n - \mu| | x_p |	\right\} 	\\
		&=& C_{p,n} \left\{ |\widehat{es}_{p,n} - es_p| (|x_p| + |\mu|) + |\hat{x}_{p,n} - x_p| ( |es_p| + |\mu|) \right.\\
		&\,\,\,& \,\,\,\,\,\,\,\,\,\,\,  \left.+ |\hat{\mu}_n - \mu|( | es_p| + |x_p|) \right\}
	\end{eqnarray*}  
	where $C_{p,n} := | (1-p)[(x_p - \mu)(\hat{x}_{p,n} - \hat{\mu}_n)]^{-1} | >0$. From the last inequality it suffices to show that the empirical estimators $\widehat{es}_{p,n}, \hat{x}_{p,n},\hat{\mu}_n$ converge almost surely to their population counterparts. We recall that $X \in \mathcal{X} \subset L^1$. We show the result in three steps:\\
	
	\noindent \emph{Step 1 (Strong consistency for $\hat{\mu}_n$)} For the empirical mean estimator (sample mean) $\hat{\mu}_n = \bar{X}$ it is well known by the Strong Law of Large Numbers (SLLN) that converges to $\mu = \E[X]$ almost surely. 
	
	\noindent \emph{Step 2 (Strong consistency for $\hat{x}_{p,n}$)} By Glivenko–Cantelli, the empirical distribution function $\hat{F}_{X,n}$ converges uniformly almost surely to $F_X$. Since $F_X$ is continuous and strictly increasing at $x_p = F_X^{-1}(p)$, this implies strong consistency of the empirical quantile $\hat{x}_{p,n}$, i.e. $\hat{x}_{p,n} \to x_{p}$ almost surely for any $p \in (0,1)$ (see e.g. Theorem 2.3.1 in \cite{serfling2009approximation}).
	
	\noindent \emph{Step 3 (Strong consistency for $\widehat{es}_{p,n}$)} We recall that by SLLN it holds that $\frac{1}{n}\sum_{i=1}^n|X_i| \to \int_{\R} |x| dF_X(x)$ almost surely. Combined with the almost surely weak convergence of the empirical induced probability measure $\mathbb{P}_{X,n}$ to the induced probability measure $\mathbb{P}_{X}$ (by Glivenko-Cantelli) it implies that $W_1(\mathbb{P}_{X,n}, \mathbb{P}_X) \to 0$ almost surely (by Proposition 7.1.5 in \cite{ambrosio2005gradient} applied with $p=1$) where $W_1$ denotes the $1$-Wasserstein distance. Since the involved probability distributions are in the real line it holds that 
	$$ W_1(\mathbb{P}_{X,n}, \mathbb{P}_X) = \int_0^1 |\widehat{F}^{-1}_{X,n}(p) - F^{-1}_X(p) | dp = \| \widehat{F}^{-1}_{X,n} - F^{-1}_X\|_1 \to 0 \quad \mbox{a.s.} $$
	Then, for any fixed $p \in (0,1)$ we have
	\begin{eqnarray*}
		|es_p - \widehat{es}_{p,n}| &=& \bigg| \frac{1}{1-p} \int_p^1 \big( \widehat{F}^{-1}_{X,n}(u) - F^{-1}_X(u) \big) du \bigg| \leq \frac{1}{1-p} \big\| \widehat{F}^{-1}_{X,n} - F^{-1}_X\big\|_1 \to 0 \quad \mbox{a.s.} 
	\end{eqnarray*}
	Combining the consistency results from the three steps we get that the empirical estimator $\hat{\theta}_{p,n}(X)$ converges almost surely to $\theta_p(X)$ for any $p \in D_X \subset (0,1)$.
\end{proof}

\begin{remark}
	As an alternative empirical estimator, one may consider a kernel-type version incorporating the output of the estimator stated in \eqref{theta-emp}, to obtain a more smooth behaviour, especially in small samples. Such an estimator in the spirit of Nadaraya-Watson nonpara\-metric local regression schemes (please see \cite{wand1994kernel}) can be easily constru\-cted. Let us denote by $\hat{\theta}_{k} := \hat{\theta}_{p_k, n}(X)$ the estimations provided by the empirical estimator at certain level points $p_k \in (0,1)$ for $k=1,2,...,m$. Then a kernel-type estimator could be constructed as
	\begin{equation*}\label{theta-np}
		\widetilde{\theta}_{p,h}(X) = \sum_{k=1}^m w_{k,h}(p) \hat{\theta}_{k}, 
	\end{equation*}
	where the local weights $\{w_{k,h}(\cdot) \}_{k=1}^m$ are calculated by
	$$ w_{k,h}(p) = \frac{ K((p-p_k)/h) }{  \sum_{\ell = 1}^m K((p-p_{\ell})/h)  }, \,\,\,\, k=1,2,...,m $$
	and $K(\cdot)$ denotes the kernel function that is used and $h>0$ denotes the smoothing parameter. The use of the Gaussian kernel function is suggested here. The choice of $h$ could be performed by applying any standard empirical rule, e.g. Silverman's rule of thumb. 
\end{remark}
	
When comparing different risks with $\theta$-index, a partial order is naturally induced. According to the aforementioned property, the resulting order concentrates on the pure shape characteristics of the risks excluding location and scale effects. The relevant definition follows. 
	
\begin{definition}\label{theta-order}
	Given two loss variables $X,Y \in \mathcal{X}$, we say that $X$ is smaller than $Y$ in the $\theta$-order, denoted by $X \leq_{\theta} Y$, if $\theta_p(X) \leq \theta_p(Y)$ for all $p \in D_X \cap D_Y$.
\end{definition}

By Proposition \ref{prop-ls-inv}, one can verify that if $X \leq_{\theta} Y$ then it also holds that $\alpha \, X + \beta \leq_{\theta} Y$ for any $\alpha>0$ and $\beta \in \R$. Therefore, the resulting stochastic order does not take into account effects that do not change the underlying tail-shape of the loss distribution. In the following, we provide a characterization result under which the discussed order is recovered.

\begin{theorem}\label{theta-tail-order}
	Consider two non-negative random variables $X,Y \in \mathcal{X}$. We have that $X \leq_{\theta} Y$ if and only if the ratio
	\begin{equation}\label{es-ratio}  
		\frac{\mbox{ES}_p(Y) - \mathbb{E}[Y]}{ \mbox{ES}_p(X) - \mathbb{E}[X]} 
	\end{equation}
is increasing in $p \in D_X \cap D_Y$.
\end{theorem}	
\begin{proof}
   Recall that $\frac{d}{dp}\left[\mbox{ES}_p(X)\right] = (\mbox{ES}_p(X) - \mbox{VaR}_p(X))/(1-p)$ and $\frac{d}{dp}\left[\mbox{ES}_p(Y)\right] = (\mbox{ES}_p(Y) - \mbox{VaR}_p(Y))/(1-p)$. Working with the required condition provided in equation \eqref{es-ratio} we have
   \begin{eqnarray*}
   	\frac{d}{dp}\bigg[ \frac{\mathrm{ES}_p(Y) - \E[Y] }{ \mathrm{ES}_p(X) - \E[X] } \bigg] &\geq& 0 \iff \\
   	( \mathrm{ES}_p(Y) - \mathrm{VaR}_p(Y) ) ( \mathrm{ES}_p(X) - \E[X] ) &\geq& ( \mathrm{ES}_p(X) - \mathrm{VaR}_p(X) ) ( \mathrm{ES}_p(Y) - \E[Y] ) \iff \\
   	\mathrm{ES}_p(Y) ( \mathrm{VaR}_p(X) - \E[X] ) + \mathrm{VaR}_p(Y) \E[X] & \geq & \mathrm{ES}_p(X) ( \mathrm{VaR}_p(Y) - \E[Y] ) + \mathrm{VaR}_p(X) \E[Y] \iff \\
   	( \mathrm{ES}_p(Y) - \mathrm{VaR}_p(Y) ) ( \mathrm{VaR}_p(X) - \E[X] ) &\geq& ( \mathrm{ES}_p(X) - \mathrm{VaR}_p(X) ) ( \mathrm{VaR}_p(Y) - \E[Y] ) \iff \\
   	\theta_p(Y) &\geq& \theta_p(X) 
   \end{eqnarray*}
\end{proof}

%%================================================================
%%                                               PARETO-TAIL CHARACTERIZATION
%%================================================================
Next, we provide some characterizations of the underlying loss variable's tail behaviour based on the formula of $\theta_p(X)$. The mean excess function (or mean residual lifetime) is employed in this attempt, since it is of great interest in actuarial science and it uniquely characterizes the loss variable's distribution (see e.g. \cite{marshall2007life} and references therein). As a particular case of interest, we consider the family of loss variables for which the mean excess loss function can be represented in an affine (linear) form, i.e. 
\begin{equation}\label{mel}
	e_X(x) = \mathbb{E}[ X - x | X > x ] =  \alpha x + \beta, \,\,\,\, x \geq 0
\end{equation}	
for $\alpha > -1, \beta > 0$. In this case, it holds tha $\mathbb{E}[X] = e_X(0) = \beta$. For any risk $X \sim$ GP($\alpha,\beta$), according to the parameterization provided in \cite{nair2013quantile}, one can verify that equation \eqref{mel} is satisfied. Since for such risks it holds that
	$$\mbox{VaR}_p(X) = \frac{\beta}{\alpha} \left[ (1-p)^{-\frac{\alpha}{\alpha+1}}-1  \right], \,\,\,\,
	\mbox{ES}_p(X) = \frac{\beta}{\alpha} \left[(\alpha + 1)(1-p)^{-\frac{\alpha}{\alpha+1}}-1 \right]$$
we can show that the related $\theta_X(p)$ is provided by the equation
\begin{equation}\label{theta-GP}
	\theta_p(X) = (1-p) \, \frac{\alpha \, \frac{\beta}{\alpha} 
		\left[(1-p)^{-\frac{\alpha}{\alpha+1}}-1\right] +\beta}
	{\frac{\beta}{\alpha} \left[(1-p)^{-\frac{\alpha}{\alpha+1}}-1 \right]- \beta},
\end{equation}
while the condition $\mbox{VaR}_p(X) > \mathbb{E}[X]$ is satisfied (given the shape parameter $\alpha$) for all $p \in (0,1)$ satisfying $\frac{1}{\alpha}( (1-p)^{-\alpha/(\alpha+1)} - 1)>0$. In the following we provide our main characterization result for the tail behaviour of a loss variable.

\begin{proposition}
	Let $X \in \mathcal{X}$ with $\E[X]=\beta$ such that  $0 \leq \beta < \infty$. Then, for any $p \in D_X$ it holds that
	\begin{equation}\label{char}
		\theta_p(X) = (1-p) \frac{ \alpha \mbox{VaR}_p(X) + \beta }{ \mbox{VaR}_p(X) - \beta }
	\end{equation}
	for all $\alpha > -1$, if and only if $e_X(\mbox{VaR}_p(X)) = \alpha \mbox{VaR}_p(X) + \beta$.
\end{proposition}	

\begin{proof}
	Assume a loss variable $X$ with $0 \leq \mathbb{E}[X] = \beta < \infty$ and $\theta_p(X)$ as stated in \eqref{char}. Then, the corresponding mean excess function at point $\mbox{VaR}_p(X)$ is given by
	$$e_X(\mbox{VaR}_p(X)) = \alpha \mbox{VaR}_p(X) + \beta$$
	which results to the mean excess function $e_X(x) = \alpha x + \beta$. 
\end{proof}
	
\color{black}	

\begin{remark}
	The Generalized Pareto distribution is a quite flexible loss model that allows for the recovery of other distributions which makes the aforementioned result quite general. For instance, $GP(\alpha,\beta)$ contains the following three well-known loss models: 
	\begin{itemize}
		\item For $\alpha \to 0$ and $\beta = 1/\lambda>0$ the exponential distribution is retrieved with tail function
		$$\overline{F}(x) = \lambda \, e^{-\lambda x}, \;\; x \geq 0.$$ 
		
		\item For $\alpha = (a-1)^{-1}> 0$ and $\beta = \kappa \, (a-1)^{-1}>0$, where $a>1$ and $\kappa>0$,  
		we get the Pareto II (or Lomax) distribution with tail function
		$$ \overline{F}(x) = \left(\frac{\kappa}{\kappa+x}\right)^a, \;\; x \geq 0.$$
		
		\item  For $\alpha = -(c-1)^{-1} \in (-1,0)$ and $\beta = \omega \, (c+1)^{-1}>0$, where $c>0$ and $\omega>0$, we get the Rescaled Beta distribution with tail function
		\[
		\overline{F}(x) = \left(1-\frac{x}{\omega}\right)^c, \;\; 0 \le x \le \omega.
		\]  
		In the special case where $c=1$ the Rescaled Beta distribution yields the Uniform distribution on the interval $[0, \omega]$.
	\end{itemize}  
\end{remark}

\color{black}

\section{Marginal risk contributions with respect to VaR via $\theta$-index}\label{Sec-4} 

We discuss here the Euler's risk allocation principle introduced in \cite{tasche2007capital} within the context of the FES and the $\theta$-index. For an aggregate loss position constituted by the sum of $d$ different loss components (differrent sectors of an insurance company), i.e.
$$ X = X_1 + X_2 + \cdots + X_d,$$
where $X\in \mathcal{X}$ and also $X_j \in \mathcal{X}$ for any $j=1,2,...,d$. The Euler's allocation principle  allows for estimating the risk contribution for each component to the total risk position. For a certain risk mapping $\varrho(\cdot)$, the $j$-th individual risk contribution is calculated by
\begin{equation}\label{euler}
	\varrho(X_j|X) = \left. \frac{d}{dh} \varrho(X + h X_j) \right|_{h=0}. 
\end{equation}
For homogeneous risk measures the full allocation property is satisfied, i.e.  
\begin{equation}\label{fap}
	\varrho(X) = \sum_{j=1}^d \varrho(X|X_j).
\end{equation}	
Under the perspective of the standard risk measures VaR and ES, the individual risk contributions are given by the formulas
\begin{eqnarray}
	\mbox{VaR}_p(X_j | X) &=& \mathbb{E}[ X_j | X = \mbox{VaR}_p(X) ] \label{VaR-euler}\\
	\mbox{ES}_p(X_j | X) &=& \mathbb{E}[ X_j | X \geq \mbox{VaR}_p(X) ] \label{ES-euler}
\end{eqnarray}	
as proved in \cite{tasche2007capital}, while it is immediate that $\mathbb{E}[X_j | X ] = \mathbb{E}[X_j]$. We turn our attention in incorporating this risk allocation approach for assessing the contri\-bution of each loss component to the aggregate risk. First, we provide a general result concerning the individual risk contributions with respect to FES. 

\begin{proposition}\label{prop-7}
	Consider an aggregate loss position $X = \sum_{j=1}^d X_j$. Assume that $X \in \mathcal{X}$ and $X_j \in \mathcal{X}$ for all $j=1,2,...,d$. Then, for any $p \in (0,1)$ and $\theta>0$, the risk contribution of the $j$-th component with respect to FES is given by
	\begin{equation}\label{fes-alloc}
		\mbox{FES}_p(X_j | X; \theta) = \frac{1-p}{1-p+\theta} \mbox{ES}_p(X_j | X) + \frac{\theta}{1-p+\theta} \mathbb{E}[X_j]
	\end{equation}
	while the full allocation property is also satisfied, i.e. $\mbox{FES}_p(X;\theta) = \sum_{j=1}^d \mbox{FES}_p(X_j | X; \theta)$.
\end{proposition}

\begin{proof}
	The individual risk allocation of FES stated in \eqref{fes-alloc} is immediate by applying the definition of the Euler's allocation in \eqref{euler} to FES representation given in \eqref{mixture}, employing the result for Euler allocation with respect to ES stated in \eqref{ES-euler}, and the immediate result $\E[X_j | X] = \E[X_j]$. The full allocation property of the risk measure, i.e. $\sum_{j=1}^d \mbox{FES}_p(X_j|X;\theta) = \mbox{FES}_p(X;\theta)$ follows immediately by the full allocation properties of ES and $\mathbb{E}[X]$.
\end{proof}

Next, we consider $\theta$-index and the resulting mixture representation of VaR through FES with level-dependent mixing weights. For notation convenience, we define this mixture as the Probability Equal Level representation of VaR through FES, i.e.
\begin{equation}\label{PELVaR}
	\mbox{PELVaR}_p(X) := \mbox{FES}_p(X; \theta_p(X)) = \mbox{VaR}_p(X)
\end{equation}
for any $p \in D_X$, with the mixture determined the same way as in \eqref{mixture} but replacing the constant parameter $\theta$ with the level-dependent functional $\theta_p(X)$ which ensures the equivalence. This allows for a level-varying weighting scheme between ES and $\E[X]$ that replicates VaR, allowing for potential insights in the inter-dependence of the involved risk measures. Next, we provide in closed form, the marginal risk contributions for $\theta$-index and VaR (through its equivalence relation with FES, i.e. PELVaR) according to the Euler's risk allocation principle. 

\begin{proposition}\label{prop-8}
	Consider an aggregate loss position $X = \sum_{j=1}^d X_j$. Assume that $X \in \mathcal{X}$ and $X_j \in \mathcal{X}$ for all $j=1,2,...,d$. Then, for any $p \in D_X$\footnote{Note that throughout this section $D_X := \bigcap_{i=1}^d D_{X_i} \subset (0,1)$ when $X$ corresponds to a sum of random variables.} the following hold:
	\begin{itemize}
		
		\item[(i)] The risk contribution of the $j$-th component to the aggregate loss position with respect to the $\theta$-index according to the Euler's allocation principle is given by
		\begin{equation}\label{theta-risk-alloc}
			\theta_p(X_j | X) = 
			\theta_p(X) \left[  \frac{\mbox{ES}_p(X_j | X) -  \mbox{VaR}_p(X_j | X)}{ \mbox{ES}_p(X) - \mbox{VaR}_p(X) } 
			- \frac{ \mbox{VaR}_p( X_j | X) - \mathbb{E}[X_j] }{ \mbox{VaR}_p(X) - \mathbb{E}[X] }    \right]	
		\end{equation}
		for $j=1,2,...,d$. 
		
		\item[(ii)] The contribution of the $j$-th component to the aggregate loss position with respect to the PELVaR according to the Euler's allocation principle, coincides with the individual risk contribution with respect to VaR and is expressed through the representation
		\begin{eqnarray}\label{fes-pel-alloc}
			\mbox{VaR}_p(X_j|X) &=& \mbox{PELVaR}_p(X_j | X) %= \mbox{FES}_p(X_j|X; \theta_p(X))
			%&=& \mbox{FES}_p(X_j | X; \theta_p(X)) \nnb \\ 
			= \frac{1-p}{1-p+\theta_p(X)} \mbox{ES}_p(X_j|X) +  \frac{\theta_p(X)}{1-p+\theta_p(X)} \mathbb{E}[X_j] \nnb \\
			&\,\,\,& - \frac{\theta_p(X_j | X)}{1-p+\theta_p(X)} ( \mbox{VaR}_p(X) - \E[X] )			
		\end{eqnarray}
		for $j=1,2,...,d$.
	\end{itemize}
\end{proposition}

\begin{proof}
	(i) It suffices to use the Euler's risk allocation definition stated in \eqref{euler} for $\theta_p(X)$. Then, we have
	\begin{eqnarray*}
		\theta_p(X_j | X) &=& \left. \frac{d}{dh}\theta_p(X + h X_j) \right\vert_{h=0} 
		= (1-p) \frac{\mbox{ES}_p(X_j | X) - \mbox{VaR}_p(X_j | X)}{ \mbox{VaR}_p(X) - \mathbb{E}[X] }  \\
		&\,\,& - (1-p) \frac{\mbox{ES}_p(X) - \mbox{VaR}_p(X) }{ (\mbox{VaR}_p(X) - \mathbb{E}[X])^2 } ( \mbox{VaR}_p(X_j | X) - \mathbb{E}[X_j])\\
		&=& \theta_p(X) \left[  \frac{\mbox{ES}_p(X_j | X) -  \mbox{VaR}_p(X_j | X)}{ \mbox{ES}_p(X) - \mbox{VaR}_p(X) } 
		- \frac{ \mbox{VaR}_p( X_j | X) - \mathbb{E}[X_j] }{ \mbox{VaR}_p(X) - \mathbb{E}[X] }     \right].
	\end{eqnarray*}
	
	\noindent (ii) Applying the definition of the Euler's risk allocation principle we get
	\begin{eqnarray*}
		\mbox{PELVaR}_p(X_j | X) %&=& \mbox{FES}_p(X_j|X; \theta_p(X)) 
		%&=& \left. \frac{d}{dh} \mbox{FES}_p(X + h X_j;\theta_p(X)) \right|_{h=0} \\
		&=&  \left. \frac{d}{dh} \mbox{PELVaR}_p(X + h X_j) \right|_{h=0} \\
		&=&  \left. \frac{d}{dh} \mbox{FES}_p(X + h X_j; \theta_p(X + h X_j)) \right|_{h=0} \\
		&=&  \left. \frac{d}{dh} \left[ \frac{(1-p)\mbox{ES}_p(X + h X_j) + \theta_p(X + h X_j) \mathbb{E}[X + h X_j]}{1-p+\theta_p(X + h X_j)} \right] \right|_{h=0} \\
		&=& \frac{1-p}{1-p+\theta_p(X)} \mbox{ES}_p(X_j | X) 
		+ \frac{\theta_p(X)}{1-p+\theta_p(X)} \mathbb{E}[X_j] \\
		&\,\,\,& - \frac{\theta_p(X_j | X)}{1-p+\theta_p(X)} \left( \frac{1-p}{1-p+\theta_p(X)}\mbox{ES}_p(X) \right. \\
		&\,\,\,& \left. + \frac{\theta_p(X)}{1-p+\theta_p(X)}\mathbb{E}[X] - \mathbb{E}[X]  \right) \\
		&=& \frac{ (1-p) \mbox{ES}_p(X_j|X) + \theta_p(X) \mathbb{E}[X_j]}{1-p+\theta_p(X)}  \\
		&\,\,\,& - \frac{ \theta_p(X_j|X) (\mbox{VaR}_p(X) - \mathbb{E}[X]) }{1-p+\theta_p(X)}\\
		&=& \frac{ (1-p)\mbox{ES}_p(X_j|X) + \theta_p(X) \E[X_j] - \theta_p(X_j|X) ( \mbox{VaR}_p(X) - \E[X] )  }{ 1-p+\theta_p(X) }.
	\end{eqnarray*}
	Since for each $p \in D_X$ we have $\mbox{PELVaR}_p(X) = \mbox{VaR}_p(X)$ as functionals on $L^1$, it follows that for any direction $Y \in L^1$ the maps $h \mapsto \mbox{PELVaR}_p(X + hY)$ and $h \mapsto \mbox{VaR}_p(X + hY)$ coincide for all $h$ near $0$. Hence, whenever the directional derivative at $h=0$ exists, it is the same for both functionals. In particular, for $Y = X_j$ we obtain
	$$ \mbox{PELVaR}_p(X_j | X) = \left. \frac{d}{dh}\mbox{PELVaR}_p(X + h X_j) \right|_{h=0} =  \left. \frac{d}{dh}\mbox{VaR}_p(X + h X_j) \right|_{h=0} = \mbox{VaR}_p(X_j | X)$$
	so the individual risk contributions coincide. Equation \eqref{fes-pel-alloc} therefore provides an explicit representation of the common Euler contributions 
	$\mbox{PELVaR}_p(X_j|X) = \mbox{VaR}_p(X_j|X)$ in terms of ES and the $\theta$-index.
\end{proof}

\begin{remark}
	Equation \eqref{theta-risk-alloc} indicates that the marginal $\theta$-contribution of component $X_j$ depends mainly on two terms: (a) the ratio $( \mbox{ES}_p(X_j|X) - \mbox{VaR}_p(X_j|X) )/(\mbox{ES}_p(X) - \mbox{VaR}_p(X))$ which measures how much loss component $X_j$ relatively contributes to the tail-accumulation gap $\mbox{ES}_p(X) - \mbox{VaR}_p(X)$, and (b) the ratio $( \mbox{VaR}_p(X_j|X) - \E[X_j] )/(\mbox{VaR}_p(X) - \E[X])$ which measure the $j$-th component relative contribution to the mean-quantile gap $\mbox{VaR}_p(X) - \E[X]$. %The outer term $\theta_p(X)$ has only a scaling effect on $\theta_p(X_j|X)$.
	Hence, it is clear that $\theta_p(X_j|X)$ is positive when $X_j$ contributes proportionally more to tail accumulation beyond VaR than to the shift from the mean, i.e. when
	$$ \frac{\mbox{ES}_p(X_j | X) -  \mbox{VaR}_p(X_j | X)}{ \mbox{ES}_p(X) - \mbox{VaR}_p(X) } 
	> \frac{ \mbox{VaR}_p( X_j | X) - \mathbb{E}[X_j] }{ \mbox{VaR}_p(X) - \mathbb{E}[X] }$$   
	or equivalently, when
	\begin{equation*}\label{alloc-ineq}
		\theta_{p,j}(X) :=  (1-p)\frac{\mbox{ES}_p(X_j | X) -  \mbox{VaR}_p(X_j | X)}{ \mbox{VaR}_p( X_j | X) - \mathbb{E}[X_j] } > \theta_p(X)
	\end{equation*}
	for any $p \in D_X$. Otherwise, $\theta_p(X_j|X)$ is negative. In this way, the $\theta$-index allocation identifies which components locally steepen or flatten the tail of $X$ at level $p$, relative to their contribution to the global gaps of the aggregate loss.
\end{remark}

\begin{remark}
	Unlike $\mbox{ES}_p(X_j | X)$, the numerical evaluation of the risk contribution $\mbox{VaR}_p(X_j | X)$ is well known to be challenging, as it requires approximating the conditional expectation $\E[X_j | X = \mbox{VaR}_p(X)]$. Regression-based and kernel-smoothing techniques have been proposed \citep{tasche2007capital, gribkova2023estimating}. These approaches typically involve choices of bandwidths or smoothig parameters and are affected by data sparsity in the tail. Proposition \ref{prop-8} provides an alternative route. Since $\mbox{PELVaR}_p(X) = \mbox{VaR}_p(X)$ for all $p\in D_X$, the marginal VaR contributions coincide with the corresponding PELVaR contributions. The latter admit the explicit representation \eqref{fes-pel-alloc}, which depends only on ES-based quantities and on the $\theta$-index allocation $\theta_p(X_j|X)$. These inputs are substantially easier to estimate, as ES allocations rely on the tail event $\{X \geq \mbox{VaR}_p(X)\}$, which has positive probability, and $\theta_p(X_j|X)$ can be obtained by simple differentiation of the $\theta$-index or via Monte Carlo approximations that are considerably more stable. Consequently, PELVaR offers an operationally attractive method for estimating VaR contributions since, instead of approximating directly the highly unstable quantity $\E[X_j | X = \mbox{VaR}_p(X)]$, one estimated the smoother objects $\mbox{ES}_p(X_j|X)$ and $\theta_p(X_j|X)$, and combines then in the aforementioned close-form identity. This suggest improved stability in practical applications.
	%\red{Numerical experiments show that $\theta_p(X_j|X)$ exhibits significantly lower variance than direct estimates of $\mbox{VaR}_p(X_j|X)$, suggesting improved stability in practical applications.}
\end{remark}

A natural question arising from the $\theta$-index allocation is whether its marginal contributions satisfy a global consistency property analogous to the full allocation identity of homogeneous risk measures. Although the $\theta$-index itself is not homogeneous, the structure revealed in \eqref{theta-risk-alloc} suggests that the contributions may exhibit a cancellation property when aggregated over all components. The next proposition shows that this is indeed the case: the marginal $\theta$-contributions always sum to zero, independently of the dependence structure or tail level $p$.

\begin{proposition}\label{prop-9}
	Consider an aggregate loss position $X = \sum_{j=1}^d X_j$. Assume that $X \in \mathcal{X}$ and $X_j \in \mathcal{X}$ for all $j=1,2,...,d$. Then, for any $p \in D_X$ it holds that
	\begin{equation}\label{theta-zero-sum}
		\sum_{j=1}^d \theta_p(X_j|X) = 0 
	\end{equation}
	inducing the full allocation property for PELVaR, i.e. $\sum_{j=1}^d \mbox{PELVaR}_p(X_j | X) %= \sum_{j=1}^d \mbox{FES}_p(X_j | X; \theta_p(X_j|X)) = \mbox{FES}_p(X;\theta_p(X)) 
	= \mbox{PELVaR}_p(X)$.
\end{proposition}

\begin{proof}
	The zero-sum risk allocation property for $\theta$-index stated in \eqref{theta-zero-sum} is immediate by combining the marginal $\theta$-index risk allocations stated in \eqref{theta-risk-alloc} with the full risk allocation properties of ES, VaR and $\mathbb{E}[X]$. Moreover, the full risk allocation property of PELVaR, is verified by summing the risk contributions with respect to PELVaR as determined in \eqref{fes-pel-alloc} and combining the full risk allocation property of FES for general $\theta$ and the zero-sum risk allocation property of $\theta$-index.
\end{proof}

This identity shows that the $\theta$-index measures relative effects across components, since their contributions must sum to zero. Together with representation \eqref{fes-pel-alloc}, it yields a coherent decomposition of VaR allocations through the $\theta$-index, thereby providing a clear and consistent way to express VaR contributions through the $\theta$-index.

We conclude this section by providing a coherent upper bound for aggregate loss positions relying on FES and the $\theta$-index.

\begin{proposition}\label{prop-risk-bound}
	Consider an aggregate loss position $X = \sum_{j=1}^d X_j$. Assume that $X \in \mathcal{X}$ and $X_j \in \mathcal{X}$ for all $j=1,2,...,d$. Then, for any $p \in D_X$ it holds that
	\begin{equation}\label{risk-sum}
		\mathrm{VaR}_p(X) \, \leq \, \sum_{i=1}^d \mathrm{FES}_p\big( X_i \, ; \, \theta_p(X) \big)\,  \leq \, \sum_{i=1}^d \mathrm{ES}_p(X_i)
	\end{equation}
\end{proposition}

\begin{proof}
	Fix $p \in D_X$ and set $\theta := \theta_p(X) > 0$. By the definition of PELVaR, we have
	$$ \mathrm{VaR}_p(X) = \mathrm{FES}_p(X;\theta). $$
	For this fixed choice of $\theta$, the mapping $Y \mapsto \mathrm{FES}_p(Y; \theta)$ is a coherent risk measure (see Prop. 3.1 in \cite{psarrakos2024flexibility}) and hence subadditive. Therefore, 
	$$  \mathrm{VaR}_p(X) = \mathrm{FES}_p(X; \theta) \leq 
			\sum_{i=1}^d \mathrm{FES}_p(X_i; \theta) = 
			\sum_{i=1}^d \mathrm{FES}_p(X_i; \theta_p(X) ).$$
	Finally, since $\mathrm{FES}_p(Y;\theta) \leq \mathrm{ES}_p(Y)$ for any $\theta>0$ and any loss position $Y$, it follows that 
	$$ \sum_{i=1}^d \mathrm{FES}_p(X_i; \theta_p(X)) \leq \sum_{i=1}^d \mathrm{ES}_p(X_i),$$
	which completes the proof.
\end{proof}

\begin{remark}
	The upper bound in Proposition \ref{prop-risk-bound} relies on the common parameter $\theta_p(X)$ associated with the aggregate loss position. As a result, the terms $\mathrm{FES}_p(X_i; \theta_p(X))$ do not coincide with the individual $\mathrm{VaR}_P(X_i)$, in contrast to the exact representation on the left-hand side. Nevertheless, this construction yields a coherent and typically less conservative upper bound than the sum of individual $\mathrm{ES}_p(X_i)$, thereby providing a meaningful compromise between the non-subadditive VaR and the overly conservative ES.
\end{remark}

%% ==================================
%%         TAIL BEHAVIOUR ANALYSIS
%% ==================================

\section{Tail behaviour analysis via $\theta$-index}\label{Sec-5}

In this section, we treat the $\theta$-index as a diagnostic tool for tail behaviour for the underlying loss distribution. In this perspective, it is more natural and convenient to work on the threshold-dependent version of the $\theta$-index. In this view, let us set for any $p \in D_X$ a threshold variable $u := \mbox{VaR}_p(X)$ providing a correspondence to both representations. Then, for any threshold $u > \mu = \E[X]$, the threshold-dependent version of the $\theta$-index of $X$ determined by 
\begin{equation}\label{theta-2}
	\theta(u) := \bar{F}_X(u) \frac{ e_X(u) }{ u - \mu }
\end{equation}
where $\bar{F}_X(u) = 1 - F_X(u)$. For what follows, we consider that $X \in \mathcal{X}$ and that the functions $\bar{F}_X(u)$ and $e_X(u)$ are strictly positive and twice differentiable at every point of the interval $(\mu, \infty)$. We turn our attention the information provided by the local derivatives of $\theta$ and $\widetilde{\theta} = \log \theta$ concerning the behaviour of the distribution of $X$, on a specified interval $[u_1, u_2] \subset (\mu, \infty)$. The logarithmic transform of $\theta$, i.e. the function $\widetilde{\theta}(u) := \log \theta(u)$ for any $u > \mu$, can be decomposed by 
$$ \widetilde{\theta}(u) = \log \theta(u) = \log \bar{F}_X(u) + \log e_X(u) - \log (u-\mu),$$
allowing for a clear distinction of the contributed term to: (a) the tail probability effect, (b) the conditional severity effect (throught the log mean excess function term), and (c) the local positional effect of the point (attachment) $u$ with respect to the main body of the distribution.  

We first study the decay rate of $\theta$ through the first derivative of $\widetilde{\theta}$. Since the derivative of $\theta$ is always negative in $(\mu, \infty)$ (by Proposition \ref{theta-decreasing}), we further examine which distributional aspects contribute to the decrease and also how rapid the decrease is. For any $u > \mu$, the log-rate of $\theta$ is given by the equation
\begin{equation}\label{theta-derivative}
	\widetilde{\theta}'(u) = \frac{d}{du}\left[ \log \theta(u) \right] = \frac{ \theta'(u) }{ \theta(u) } = -\left( \frac{1}{e_X(u)} + \frac{1}{u-\mu} \right) < 0.
\end{equation}
The expression \eqref{theta-derivative} reveals that the percentage rate at which tail heaviness diminishes when the attachment increases from $u$ to $u + du$ for infinitesimal $du>0$, is determined by the inverse mean excess function and the inverse distance from the mean $\mu$ of the distribution. From the actuarial point of view, a large mean excess (i.e. deep heavy tail) implies a small $1/e_X(u)$ as $u\to \infty$ and therefore a slower decay rate of $\theta$, indicating that tail heaviness persists. On the other hand, a small or flattening mean excess implies a large $1/e_X(u)$ as $u$ increases, leading to a rapid decay rate of $\theta$, indicating that the tail becomes attritional beyond the threshold $u$.

It is more informative to understand how the $\theta$-index behaves across an attachment interval. This is especially relevant in an actuarial context, where insurance and reinsurance layers correspond to finite intervals $[u_1,u_2]$ rather than isolated points and is important for the optimal design of the layers of the contracts. In particular, one would like to quantify how much of the tail heaviness observed at a lower attachment $u_1$ persists when the attachment is increased to $u_2$, and how this persistence depends on the characteristics of the mean excess function within the layer. The following result provides such an interval-level comparison by integrating the log-rate of $\theta$ across $[u_1, u_2]$, yielding a clean expression and practical bounds for the ratio $\theta(u_1)/\theta(u_2)$.

\begin{proposition}\label{prop-10}
	Let $\mu < u_1 < u_2 < \infty$. Then, we have that 
	\begin{equation}\label{dec-1}
		\frac{ \theta(u_1) }{ \theta(u_2) } =  \exp\bigg\{ \int_{u_1}^{u_2} \left[ \frac{1}{e_X(u)} + \frac{1}{u-\mu}\right] du \bigg\}.
	\end{equation}
	If in addition, $0 < m \leq e_X(u) \leq M < \infty$ for any $u \in [u_1, u_2]$ with 
	$ m:= \inf_{u \in [u_1,u_2]} e_X(u)$ and $M := \sup_{u \in [u_1,u_2]} e_X(u)$, we get the bounds
	\begin{equation}\label{dec-2}
		\frac{u_2 - \mu}{ u_1 - \mu } \, e^{ (u_2-u_1)/M } \leq \frac{ \theta(u_1) }{ \theta(u_2) } \leq 	\frac{u_2 - \mu}{ u_1 - \mu } \, e^{ (u_2-u_1)/m } 
	\end{equation}
\end{proposition}

\color{black}

\begin{proof}	
	The integral representation stated in \eqref{dec-1} is immediate by integrating \eqref{theta-derivative} on the interval $[u_1, u_2]$. For the inequality in \eqref{dec-2}, working on the interval $[u_1, u_2]$, and using that $m \leq e_X(u) \leq M$ we get
	\begin{eqnarray*}
		m \leq e_X(u) \leq M \Rightarrow \frac{1}{M} + \frac{1}{u-\mu} \leq \frac{1}{e_X(u)} + \frac{1}{u-\mu} \leq \frac{1}{m} + \frac{1}{u-\mu}.
	\end{eqnarray*} 
	Integrating on $[u_1, u_2]$ and multiplying we get
	\begin{eqnarray*}
		\int_{u_1}^{u_2} \left( \frac{1}{M} + \frac{1}{u-\mu} \right) du \leq  \int_{u_1}^{u_2} \left( \frac{1}{e_X(u)} + \frac{1}{u-\mu} \right) du \leq  \int_{u_1}^{u_2} \left( \frac{1}{m} + \frac{1}{u-\mu} \right) du \Rightarrow \\
			\int_{u_1}^{u_2} \left( \frac{1}{M} + \frac{1}{u-\mu} \right) du \leq \log\frac{ \theta(u_1)}{ \theta(u_2) } \leq  \int_{u_1}^{u_2} \left( \frac{1}{m} + \frac{1}{u-\mu} \right) du.
	\end{eqnarray*}
	The outer integrals calculation gives
	$$ \int_{u_1}^{u_2} \frac{1}{u-\mu}du = \log(u_2 - \mu) - \log(u_1 - \mu) = \log \frac{u_2 - \mu}{ u_1 - \mu},$$
	and replacing above we get
	\begin{eqnarray*}
		&& \frac{u_2-u_1}{M} +  \log \frac{u_2 - \mu}{ u_1 - \mu} \leq \log \frac{\theta(u_1)}{ \theta(u_2) } \leq \frac{u_2-u_1}{m} +  \log \frac{u_2 - \mu}{ u_1 - \mu}  \Rightarrow \\
		%&&\exp\left[ -\left( \frac{u_2-u_1}{m} +  \log \frac{u_2 - \mu}{ u_1 - \mu} \right) \right] \leq \frac{\theta(u_1)}{ \theta(u_2) } \leq \exp\left[ -\left( \frac{u_2-u_1}{M} +  \log \frac{u_2 - \mu}{ u_1 - \mu} \right) \right] \Rightarrow \\
		&&\frac{u_2 - \mu}{u_1 - \mu} e^{(u_2-u_1)/M} \leq   \frac{\theta(u_1)}{ \theta(u_2) } \leq \frac{u_2 - \mu}{u_1 - \mu} e^{(u_2-u_1)/m}
	\end{eqnarray*}
	which concludes the proof.
\end{proof}

The representation \eqref{dec-1} expresses the relative decay of the $\theta$-index across a layer $[u_1,u_2]$. From an actuarial viewpoint, this shows that the persistence of tail heaviness between two attachments is governed jointly by the shape of the mean excess function and the relative position of the layer with respect to the body of the distribution. The bounds in \eqref{dec-2} further quantify this behaviour in a model-robust way: if the mean excess is uniformly large on $[u_1,u_2]$,  then the exponential term is close to one, indicating that tail heaviness decays only slowly across the interval, a feature that characterizes very heavy-tailed regimes. Conversely, if the mean excess is uniformly small or nearly flat, then the exponential terms collapse rapidly, implying that tail heaviness diminishes quickly as we move up the layer. These bounds therefore offer a structural diagnostic for the behaviour of insurance layers, providing insights into how sensitive tail severity is to changes in attachment level and offering a robust criterion to distinguish persistent, homogeneous, and tapering regions/regimes of the tail.

While Proposition \ref{prop-10} quantifies how much tail heaviness survives across a layer, the next step is to determine how this characteristic changes locally, that is whether tail heaviness decays at an accelerating or decelerating rate. This offers a natural way to identify tail regimes across operational layers. It is easy to verify that the curvature of $\widetilde{\theta}$ yields for any $u \in I$ the decomposition 
\begin{equation}\label{curvature}
	\widetilde{\theta}''(u) =  \frac{e'_X(u)}{e_X^2(u)} + \frac{1}{(u - \mu)^2} = \frac{ h(u) e_X(u) - 1}{e_X^2(u)} + \frac{1}{(u - \mu)^2}
\end{equation}
where $h(\cdot)$ denotes the hazard rate function of $X$. 

%Following classical reliability theory, $X$ is said to have increasing (IMRL), decreasing (DMRL) or constant (CMRL) mean residual life (equivalently, mean excess function) if $e_X(u)$ is respectively non-decreasing, non-increasing or constant in $u$ \citep{barlow1975statistical, belzunce2015introduction, shaked2007stochastic}. The CMRL property is standard and is known to characterize uniquely the exponential distribution. Based on these characterizations, 

Our next result performs a classification of the tail regimes by studying the curvature of $\widetilde{\theta}$.

\begin{proposition}\label{prop-11}
	Assume an interval $[u_1, u_2] \subset (\mu, \infty)$ and $\theta \in C^2([u_1, u_2])$. The tail regimes can be distinguished in the following cases:
	\begin{itemize}
		\item[(a)] The function $e_X(u)$ is strictly increasing in $u \in [u_1, u_2]$
		%The distribution of $X$ presents locally IMRL behaviour in $J$ 
		if and only if $\widetilde{\theta}''(u) > (u-\mu)^{-2}$ for any $u \in [u_1, u_2]$.
		%$u \in J$ (IMRL regime). 
		
		\item[(b)] The function $e_X(u)$ is strictly decreasing in $u \in [u_1, u_2]$
		if and only if $\widetilde{\theta}''(u) < (u-\mu)^{-2}$ for any $u \in [u_1, u_2]$.
		%The distribution of $X$ presents locally DMRL behaviour in $J$ if and only if $\widetilde{\theta}''(u) < (u-\mu)^{-2}$ for any $u \in J$ (DMRL regime).
		
		\item[(c)] The function $e_X(u)$ is constant in $u \in [u_1, u_2]$
		if and only if $\widetilde{\theta}''(u) = (u-\mu)^{-2}$ for any $u \in [u_1, u_2]$.
		%The distribution of $X$ presents locally CMRL behaviour in $J$ if and only if $\widetilde{\theta}''(u) = (u-\mu)^{-2}$ for any $u \in J$ (CMRL regime).
	\end{itemize}
\end{proposition}

\begin{proof} From the calculation of the curvature of $\widetilde{\theta}$ we have that:
\begin{eqnarray*}
	\widetilde{\theta}''(u) = \frac{ e'_X(u) }{ [e_X(u)]^2 }  + \frac{1}{(u-\mu)^2}
\end{eqnarray*}
or equivalently
\begin{equation}\label{ref-rel}	
	e'_X(u) = [e_X(u)]^2 \, (\widetilde{\theta}''(u) - (u-\mu)^{-2}).
\end{equation}
The results stated in (a),(b),(c) follow directly by \eqref{ref-rel}.
%(a) Since the IMRL behaviour is identified by the condition $e_X'(u) >0$, from the equation above this is equivalent to the condition $\widetilde{\theta}''(u) > (u-\mu)^{-2}$. \\
%(b) The DMRL behaviour is identified by the condition $e_X'(u) <0$, so from the equation above this is equivalent to the condition $\widetilde{\theta}''(u) < (u-\mu)^{-2}$. \\
%(c) The CMRL behaviour is identified by the condition $e_X'(u) = 0$, so from the equation above this is equivalent to the condition $\widetilde{\theta}''(u) = (u-\mu)^{-2}$. 
\end{proof}

\begin{remark}\label{rmk-5}
	The conditions of Proposition \ref{prop-11} are formulated in terms of strict inequalities. In empirical settings, where no parametric distributional form is assumed, it is common to observe regions of the tail in which the empirical mean excess behaviour fluctuates around the boundary between the cases identified in the proposition. This corresponds to situations in which
	$$ |\widetilde{\theta}''(u) - (u-\mu)^{-2}| \leq \epsilon $$
	for some $\varepsilon >0$. If this condition holds for any $u$ in an interval $[u_1,u_2]$, we refer to this layer as a \emph{transition regime}, indicating a range of loss levels over which tail behaviour is not yet clearly aligned with a single asymptotic case of Proposition \ref{prop-11}, but instead lies close to the corresponding boundary. In practice, such transition regions often correspond to ranges where the empirical tail behaviour appears approximately stable, in the sense that it remains close to the constant mean excess boundary over a non-negligible interval.
\end{remark}

The derivative and curvature of $\log \theta$ provide a practical framework for examining how tail behaviour varies across different loss levels. In insurance and reinsurance applications, attachment levels typically span an interval rather than a single point, and the local geometry of $\theta$ reflects the behaviour of losses within this range. Changes in the sign of $(\log \theta)''(u) - (u-\mu)^{-2}$, indicate transitions between distinct tail behaviours. At lower levels, tail variation is mainly driven by the magnitude of losses, at higher levels the tail exhibits faster decay while between these extremes, both effects are present. These transitions are relevant for understanding the evolution of risk across layers and the sensitivity of pricing to changes in attachment. The results of this section therefore provide a natural basis for subsequent empirical analysis, in which the behaviour of $\theta$ and its local derivatives (decay rate, curvature) can be used to detect and interpret such changes in data.

%%% =====================================================
%%%                                             APPLICATIONS
%%% =====================================================

\section{Illustration of $\theta$-index for standard loss models}\label{Sec-6}

In this section, we study some examples of standard loss models in the actuarial practice, providing $\theta$-index in closed form (wherever is possible) and determining the set $D_X$. 
%Given the location-scale invariance property of $\theta$-index (stated in Proposition \ref{prop-ls-inv}), the obtained expressions are independent of the location-scale features of the distribution. 
%Moreover, illustrations are provided for comparison of the $\theta$-index behaviour within and across the distri\-bution families considered.
%\subsection*{Some cases of shape invariant loss distributions}
First we examine some standard cases of shape invariant loss models, i.e. distributions which parameterization does not affect the underlying shape of the distribution but only location and scale features. Standard distributions that display this property are the Uniform, Normal and Exponential. These three cases provide some interesting benchmarks in the perspective of $\theta$-index, and could possibly be employed for a rough distinction with respect to the tail behaviour.

%% ============
%% UNIFORM
%% ============	
\begin{example}[Uniform]
	The simpliest case of a loss variable is when the Uniform distribution is considered (flat risk), i.e. $X \sim \mathcal{U}([\alpha, \beta])$ with $\alpha < \beta$ and $\alpha+\beta \geq 0$. Following the scale invariance property of $\theta$-index, we obtain the calculation
	\begin{equation*}\label{theta-unif}
		\theta_{\scriptsize \mbox{Unif}}(p) = \frac{ (1-p)^2 }{ 2p-1 }
	\end{equation*}
	which is independent of the location characteristics of the distribution. Taking into account the symmetry of this distribution, we obtain $D_X = (0.5, 1)$.
\end{example}

%% ========
%% NORMAL
%% ========
\begin{example}[Normal]
	For a Normal distributed loss variable $X \sim N(\mu, \sigma^2)$ with $\mu>0$, we have that 
	\begin{equation*}\label{theta-normal}
		\theta_{\scriptsize \mbox{Normal}}(p) = \frac{\varphi(\Phi^{-1}(p))}{ \Phi^{-1}(p) } - (1-p)
	\end{equation*}
	where $\varphi(\cdot)$ and $\Phi^{-1}(\cdot)$ denote the probability density function and the quantile function of the standard Normal distribution while it is easy to verify that $D_X = (0.5, 1)$.
\end{example}

%% ============
%% EXPONENTIAL
%% ============
\begin{example}[Exponential]
	For a loss random variable $X \sim$ Exp($\lambda$) with scale parameter $\lambda>0$ and  distribution $ F(x) =1- e^{-\lambda x}$ for $x \geq  0$, we can easily verify that
	\begin{equation*}\label{theta-exp}
		\theta_{\scriptsize \mbox{Exp}}(p) = -\frac{1-p}{\log(1-p)+1}.
	\end{equation*}
	and $D_X = (1-1/e, 1)$.
\end{example}

\begin{figure}[ht!]
	\centering
	\includegraphics[width=2.8in]{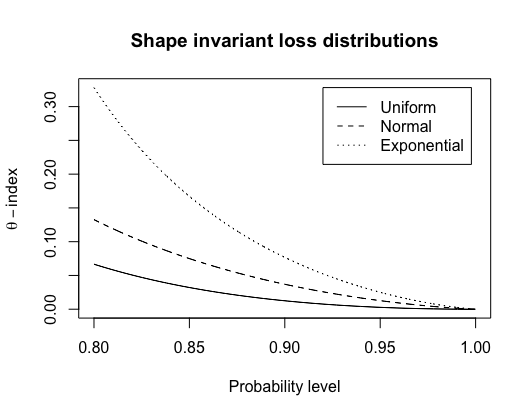}
	\caption{Illustration of $\theta_p(X)$ for loss distributions that maintain their shape pattern.}\label{fig-1}
\end{figure}

In Figure \ref{fig-1} are illustrated together the $\theta$-curves for the aforementioned basic parametric models. Among these three models, the Uniform distribution displays the less dangerous tail behaviour (fastest tail decay), while the exponential distribution displays the slowest tail decay. 

%The exponential model is a benchmark for the  not considered as a heavy tail distribution, it can be actually used as a benchmark to distinguish if a loss distribution provides heavy tail characteristics. In this view, $\theta_{\mbox{\scriptsize Exp}}$ provides a natural (normalized) bound for classifying the pure tail risk of a loss distribution (i.e. excluding location and scale effects) as one of a light-tailed or heavy-tailed status, checking if $\theta$-index for a loss distribution lies (systematicaly) below or above the $\theta_{\mbox{\scriptsize Exp}}$ curve. 
%\subsection*{Some cases of loss distributions with varying shapes}
Next, we examine some standard loss distributions with varying shape features which are often employed in the actuarial and reliability theory and practice. From this family we examine  Student-t, LogNormal, Gamma, Weibull, Pareto and the more flexible Generalized Extreme Value (GEV) model.

%% ===========
%% STUDENT-t
%% ===========
\begin{example}[Student-t]
	A $t$-distributed loss random variable displays quite similar behaviour with the Normal distribution, but allows for heavier tail controlled by the degrees of freedom (shape) parameter $\nu \geq 1$. In this case, the $\theta$-index is given by 
	\begin{equation*}\label{theta-student}
		\theta_t(p \, ;\nu) = \frac{ g_{\nu}(t_{\nu}^{-1}(p)) (\nu + (t_{\nu}^{-1}(p))^2)  }{ (\nu-1) t_{\nu}^{-1}(p) } - (1-p)
	\end{equation*}
	where $g_{\nu}(\cdot)$ and $t^{-1}_{\nu}(\cdot)$ denote the probability density function and the quantile function of the standard Student-t distribution with $\nu$ degrees of freedom. Because of the symmetry of the distribution we obtain $D_X = (0.5,1)$ which is independent of the shape parameter $\nu$.
\end{example}

%% =======================
%% LOGNORMAL DISTRIBUTION
%% =======================
\begin{example}[LogNormal]
	The LogNormal distribution is a standard model in risk theory and reliability analysis. For $X \sim LN(\mu, \sigma)$ the $\theta$-index depends only on the shape parameter $\sigma>0$, and after some algebra one can derive the formula
	\begin{equation*}\label{theta-LN}
		\theta_{\mbox{\scriptsize LN}}(p \, ;\sigma) = \frac{\left( e^{\sigma^2/2} \bar{\Phi}(\Phi^{-1}(p) - \sigma) - (1-p) e^{\sigma \Phi^{-1}(p)}  \right)}{ e^{\sigma \Phi^{-1}(p)} - e^{\sigma^2/2} }  
	\end{equation*}
	where $\bar{\Phi}(\cdot) = 1 - \Phi(\cdot)$. Moreover, it is easy to verify that $D_X = (\Phi( \frac{\sigma}{2}),1)$.
\end{example}

%%==========================
%%                   GAMMA
%%==========================
\begin{example}[Gamma]
	Consider a loss random variable $X \sim$ Gamma$(\alpha, \lambda)$ with distribution function given by
	$$ F(x) =  \frac{ \gamma(\alpha, \lambda x) }{ \Gamma(\alpha) }, \,\,\, x \geq 0, $$
	where $\gamma(s,y) = \int_0^y t^{s-1} e^{-t} dt$ denotes the lower incomplete gamma function, $\Gamma(s) = \int_0^{\infty} t^{s-1} e^{-t} dt$ denotes the gamma function and $\alpha>0, \lambda>0$ represent the shape and scale parameters, respectively. Then, it is obtained the semi-explicit expression for $\theta$-index
	\begin{equation*}\label{theta-gamma}
		\theta_{\mbox{\scriptsize Gamma}}(p \, ; \alpha) = \frac{ \int_p^1 \gamma^{-1}(\alpha, \Gamma(\alpha) s ) ds - (1-p) \gamma^{-1}(\alpha, \Gamma(\alpha) p)  }{\gamma^{-1}(\alpha, \Gamma(\alpha) p) - \alpha}
	\end{equation*}
	where $D_X =  \left(\frac{\gamma(\alpha, \alpha)}{ \Gamma(\alpha) }, 1\right)$ depends on the shape parameter $\alpha>0$.
\end{example}

%% ==============
%%          WEIBULL
%% ==============
\begin{example}[Weibull]
	Consider a loss random variable $X \sim$ Weibull$(\alpha, \lambda)$ with shape parameter $\alpha>0$, scale parameter $\lambda>0$ and distribution function given by $ F(x) = 1-\alpha \lambda (\lambda x)^{\alpha - 1} e^{-(\lambda x)^{\alpha}}$ for $x\geq 0$. Then, the $\theta$-index is obtained in semi-closed form as
	\begin{equation*}\label{theta-wb}
		\theta_{\mbox{\scriptsize Weibull}}(p \, ; \alpha) = \frac{ \int_p^1 (-\log(1-s))^{1/\alpha} ds - (-\log(1-p))^{1/\alpha} }{ (-\log(1-p))^{1/\alpha} - \Gamma(1 + 1/\alpha) }
	\end{equation*}
	and $D_X = \left( 1 - e^{ - \left(\Gamma\left(1 + \frac{1}{\alpha}\right) \right)^{\alpha}}, 1\right)$.
\end{example}

%% ==============
%% PARETO (LOMAX)
%% ==============
\begin{example}[Pareto II or Lomax]% (or Lomax distribution)}
For a loss random variable $X \sim$ Pareto$(\alpha, \kappa)$ with distribution function 
$ F(x) = 1 - [\kappa/(\kappa + x)]^{\alpha}$ for $x > 0$, scale parameter $\kappa>0$ and shape parameter $\alpha > 1$, the $\theta$-index is obtained in closed form as
\begin{equation*}\label{theta-pareto}
	\theta_{\mbox{\scriptsize Pareto}}(p \, ; \alpha) = \frac{1-p}{(\alpha-1) - \alpha(1-p)^{1/\alpha}}
\end{equation*}
and $D_X = \left(  1 - \left( \frac{\alpha-1}{\alpha}\right)^{\alpha}, 1\right)$.
\end{example}

%% ==============
%%           GEV
%% ==============
\begin{example}[Generalized Extreme Value (GEV) distribution]
For a loss random variable $X \sim$ GEV$(\mu, \sigma, \xi)$ with $\mu$ the location parameter, $\sigma>0$ the scale parameter and $\xi$ the shape parameter, the distribution function is defined as
\begin{equation*}
	F(x) = \left\{
	\begin{array}{ll}
		\exp\left\{  -\exp\left( - \frac{x-\mu}{\sigma} \right)  \right\}, & \xi = 0\\
		\exp\left\{ -\left( 1 + \xi \frac{x - \mu}{\sigma}  \right)^{-1/\xi} \right\}, & \xi \ne 0
	\end{array}	
	\right.
\end{equation*}	
The $\theta$-index can be written in semi-closed form as
\begin{equation*}\label{theta-gev}	
	\theta_{\mbox{\scriptsize GEV}}(p \, ; \xi) = \left\{
	\begin{array}{ll}
		\frac{ \gamma \left(1-\xi, -\log(p) \right) - (1-p) \left( -\log(p) \right)^{-\xi} }{ \left( -\log(p) \right)^{-\xi} - \Gamma(1-\xi) }, & \xi \ne 0, \,\, \xi < 1\\
		\frac{ \mbox{li}(p) } { \log( -\log(p) ) + \gamma_E } - 1,
	& \xi = 0\\
	0, & \xi \geq 1,
\end{array}
\right.
\end{equation*}	
where $\gamma_{E}$ denotes the Euler's constant ($\simeq 0.5772$) and $\mbox{li}(x) := \int_0^x (\log(t))^{-1} dt$ denotes the logarithmic integral function, while the set $D_X$ depends on the shape parameter, and in particular is determined by
\begin{equation*}
D_X = \left\{ 
\begin{array}{ll}
	\left( e^{ -e^{-\gamma_E}}, 1 \right), & \xi = 0 \\
	\left(e^{-(\Gamma(1-\xi))^{-1/\xi }}, 1 \right), & \xi \in \R \setminus \{0\}
\end{array}	
\right.
\end{equation*}
\end{example}

%%%%

\begin{table}[ht!]
\centering
\resizebox{\textwidth}{!}{
\begin{tabular}{c|ccc|ccc|ccc|ccc}
	\hline
	& Exp 
	& Normal
	& Uniform
	& \multicolumn{3}{c|}{Student-t} 
	& \multicolumn{3}{c|}{LogNormal} 
	& \multicolumn{3}{c}{Weibull}\\
	p & & & & \multicolumn{3}{c|}{$\nu$} &\multicolumn{3}{c|}{$\sigma$} &\multicolumn{3}{c}{$\alpha$}  \\ 
	& & &  &2 & 4 & 20 & 0.2 & 0.5 & 1.00 & 0.75 & 1.5 & 10 \\
	\hline
	0.900 & 
	0.0767&                                        %exp
	0.0369&                                       %normal
	0.0125&                                        %uniform
	0.1250&  0.0630& 0.0406&      %student
	0.0490&  0.0737& 0.1440&       %LN	
	0.1064&  0.0541& 0.0267\\        %weibull	
	
	0.950 & 
	0.0250& 
	0.0127 & 
	0.0028&
	0.0555& 0.0251& 0.0144& %student
	0.0170&  0.0255& 0.0478&
	0.0336& 0.0181& 0.0091\\ 
	
	0.975 & 
	0.0092&  
	0.0048& 
	0.0007&
	0.0263& 0.0110& 0.0056&  %student
	0.0065&  0.0099& 0.0184&
	0.0123&  0.0068& 0.0034\\
	
	0.990 & 
	0.0027&  
	0.0014&
	0.0001&
	0.0102& 0.0039& 0.0018& %student		
	0.0020& 0.0031& 0.0058&
	0.0036& 0.0020& 0.0010\\ 
	
	0.995 & 
	0.0011 &  
	0.0006&
	0.0000&
	0.0050& 0.0019& 0.0008& %student		
	0.0008& 0.0013& 0.0025&
	0.0015&  0.0008& 0.0004\\

	\hline
	\multicolumn{10}{c}{}
\end{tabular}	}

\resizebox{\textwidth}{!}{
\begin{tabular}{c|cccc|cccc|cccc}
	\hline
	& \multicolumn{4}{c|}{Gamma} 
	%& \multicolumn{4}{c|}{Generalized Pareto} 
	& \multicolumn{4}{c|}{Generalized Extreme Value}
	& \multicolumn{4}{c}{Pareto II (Lomax)}\\
	p & \multicolumn{4}{c|}{$\alpha$} 
	& \multicolumn{4}{c|}{$\xi$}
	%&  \multicolumn{4}{c|}{$\alpha$} 
	& \multicolumn{4}{c}{$\alpha$}\\ 
	&  0.25 & 0.5 & 1.5 & 20 
	%& -0.5 & -0.1 & 0.5 & 1.0 
	& -1 & 0 & 0.2 & 0.4
	&1.5 & 2 & 4 & 10 \\ 
	\hline
	0.900 & 
	0.1422 & 0.0989& 0.0684&  0.0500&     %Gamma
	0.0059 & 0.0613& 0.0998& 0.1745 &      %GEV
	%0.0012 & 0.0616& 0.1646&   0.2721&		%GP
	0.5655 & 0.2721 & 0.1332& 0.0946\\                       % Pareto                
	
	0.950 & 
	0.0398& 0.0306& 0.0227& 0.0155& %Gamma
	0.0014& 0.0212& 0.0355& 0.0620&
	%0.0028& 0.0196& 0.0559& 0.0904&
	0.1687 & 0.0555& 0.0164&  0.0138\\ 
	
	0.975 & 
	0.0137&  0.0110& 0.0085& 0.0059& %Gamma
	0.0003& 0.0081& 0.0142& 0.0255&
	%0.0006& 0.0070& 0.0222& 0.0366&
	0.0672 & 0.0366& 0.0177&  0.0120\\
	
	0.990 & 
	0.0038 & 0.0032& 0.0026& 0.0018& %Gamma
	0.0001  & 0.0025& 0.0047& 0.0088&
	%0.0001& 0.0020& 0.0074& 0.0125&  
	0.0232 & 0.0125& 0.0058& 0.0037\\  
	
	0.995 & 
	0.0015& 0.0013 &  0.0011 &  0.0008& %Gamma
	0.0000& 0.0011 &  0.0021 &  0.0041&
	%0.0000& 0.0008& 0.0034& 0.0058&
	0.0109 & 0.0058& 0.0026& 0.0016\\
	
	\hline
\end{tabular}	}
\caption{$\theta$-index for loss distributions with different shape features at upper tail levels.}\label{tab-1}
\end{table}

\begin{figure}[ht!]
\centering
\includegraphics[width=2in]{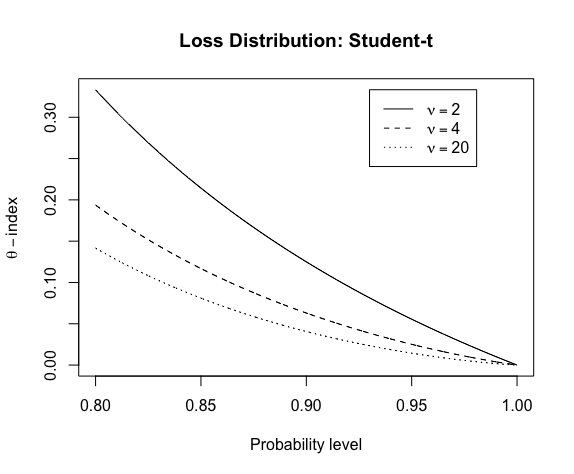}
\includegraphics[width=2in]{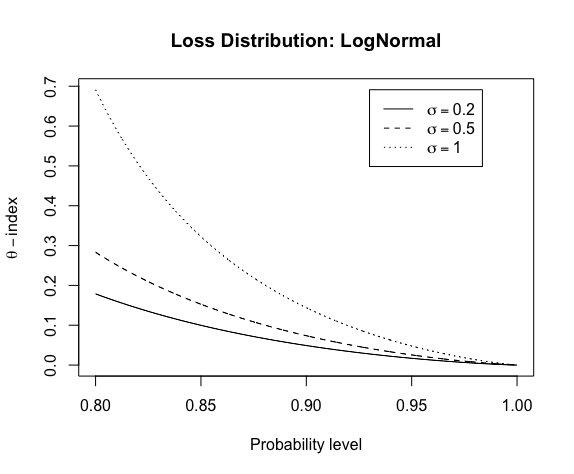}
\includegraphics[width=2in]{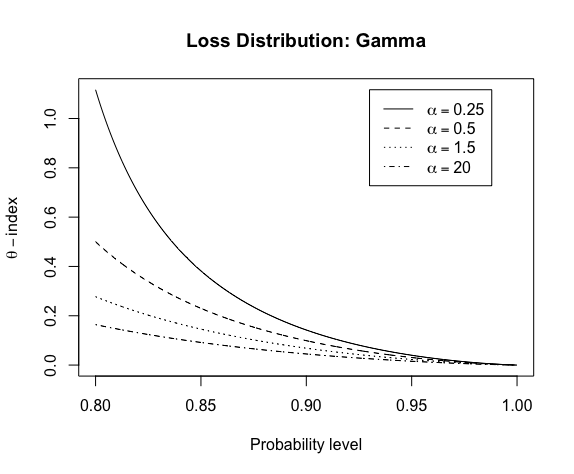}
\includegraphics[width=2in]{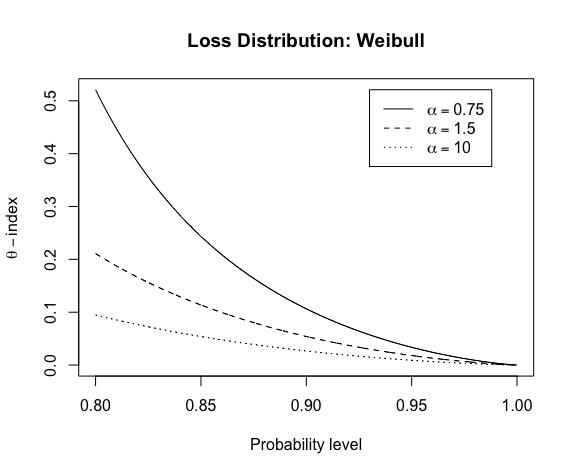}
\includegraphics[width=2in]{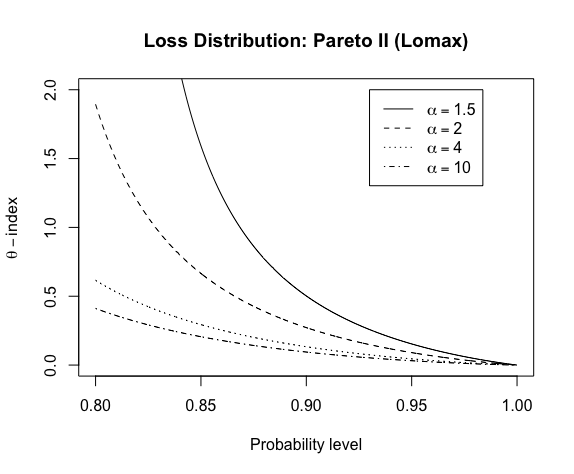}
\includegraphics[width=2in]{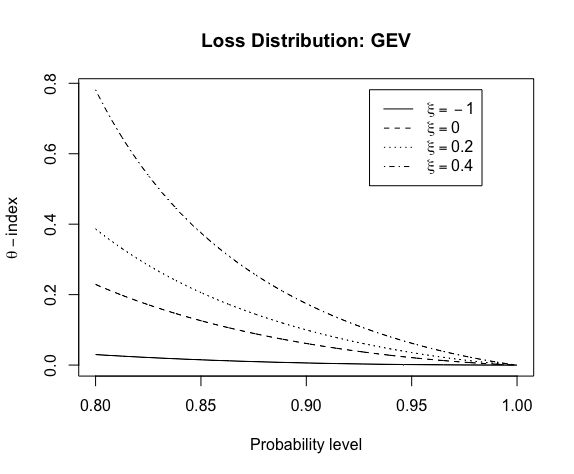}
\caption{The $\theta$-index for some shape-varying loss distributions.}\label{fig-2}
\end{figure}

In Table \ref{tab-1} are concentrated the $\theta$-index values at certain upper level points for all loss distributions considered. Moreover, in Figure \ref{fig-2} are illustrated the relevant $\theta$-index curves for $p \in (0.9,1)$. Within the same distribution family, it is clear from the reported results in the table and the plots that the shape parameter determines the ordering relation in terms of the tail shape as quantified by the $\theta$-index ($\theta$-order). For instance, for the Pareto II (Lomax) case, as the shape parameter $\alpha$ approaches to 1, we obtain higher values for $\theta$-index at any level $p$, while as $\alpha$ grows lower values are displayed, i.e. the tail decay rate becomes higher. Observe, that this behaviour remains consistent within any distribution, since the relevant shape parameter is connected through a monotone relation to $\theta$-index, and therefore parameter values that indicate potential heavy-tail behaviour will lead to higher tail risk assessment. Note that this consistency is not necessarily observed across different distribution families. 

\begin{figure}[ht!]
	\centering
	\includegraphics[width=2in]{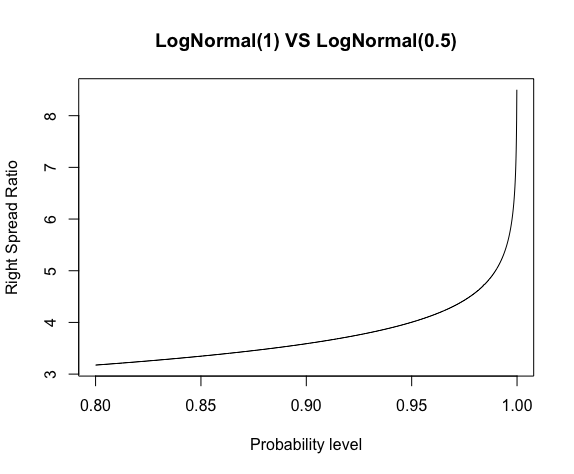}
	\includegraphics[width=2in]{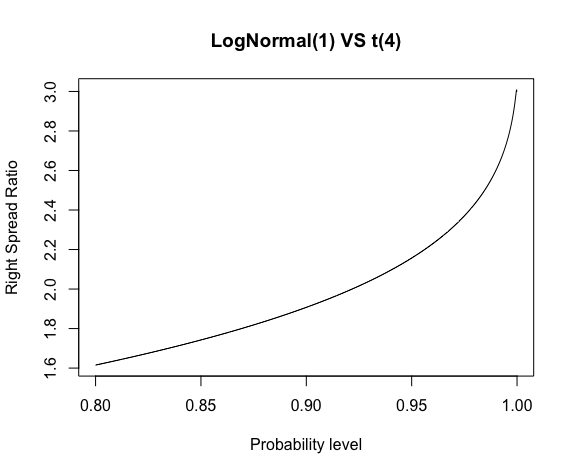}
	\includegraphics[width=2in]{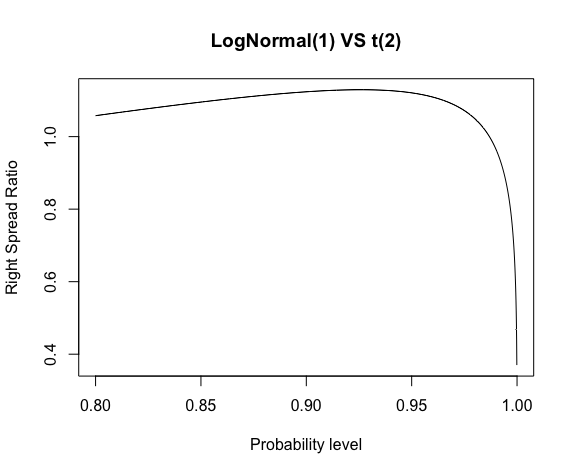}
	\caption{Illustration of the right spread ratio curve for couples of different risks}\label{fig-3}
\end{figure}

Check for instance, in Table \ref{tab-1} the LogNormal case for $\sigma = 1$ and Student-t for $\nu=2$ at levels $p=90\%$ and $p=95\%$ where the provided ordering differs. However, $\theta$-index itself, could assist standard risk measures like VaR which do not take into account the tail risk for comparison of different risk profiles. For instance, between two risks $X, Y$ for which at some level $p$ it holds that $\mbox{VaR}_p(X) = \mbox{VaR}_p(Y)$, we might consider as more dangerous $Y$ if $\theta_p(Y) > \theta_p(X)$. However this assessment provides only a local ordering of the relevant tail risks, since for a different level $p' > p$ it is not necessary that the inequality $\theta_{p'}(Y) \geq \theta_{p'}(X)$ holds. This aligns with Theorem \ref{theta-tail-order} according to which the $\theta$-ordering is satisfied if and only if the right spread ratio $(\mbox{ES}_p(Y) - \E[Y])/(\mbox{ES}_p(X) - \E[X])$ remains increasing for any $p \in D_X \cap D_Y$. In Figure \ref{fig-3} we indicatively illustrate the right spread ratio curve for some cases. First, the ratio between two members of the LogNormal family is illustrated, in which the ordering is clear from \ref{fig-2} since $\theta_{LN}(p;1) > \theta_{LN}(p;0.5)$, and consequently, the relevant right spread ratio remains increasing as indicated in Figure \ref{fig-3}. For the case where $X \sim t_4$ and $Y \sim LN(1)$, it seems that the relevant ratio remains increasing, leading to the conclusion that $X \leq_{\theta} Y$ for this case. However, this is not the case when $X \sim t_2$ and $Y \sim LN(1)$, since the right spread ratio does not remain increasing for all $p \in D_X \cap D_Y$, and therefore the order $X \leq_{\theta} Y$ does not hold.

\section{Empirical tail analysis of the Danish fire dataset}\label{Sec-7}

 In this section, we implement the proposed tail analysis framework using the dataset \emph{danishuni} from the R package \emph{CASdatasets}. The data consist of 2167 fire insurance losses collected by Copenhagen Reinsurance over the period 1980–1990, with claim amounts adjusted for inflation to 1985 price levels and expressed in millions of Danish kroner. This dataset is widely used in actuarial and extreme-value analysis and is well known for exhibiting heavy-tailed behaviour.
 
 \begin{figure}[ht!]
 	\centering
 	\includegraphics[width=2.2in]{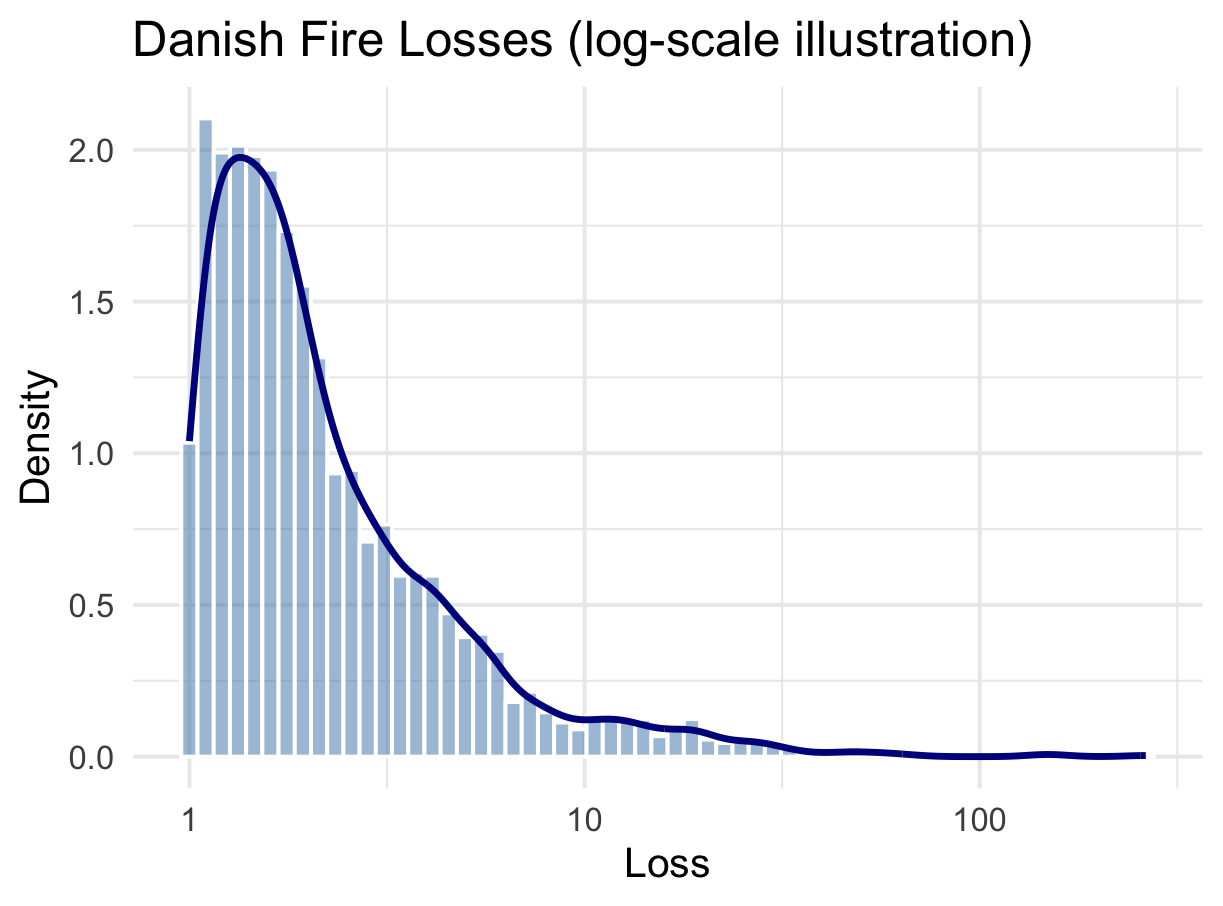}
 	\includegraphics[width=2.2in]{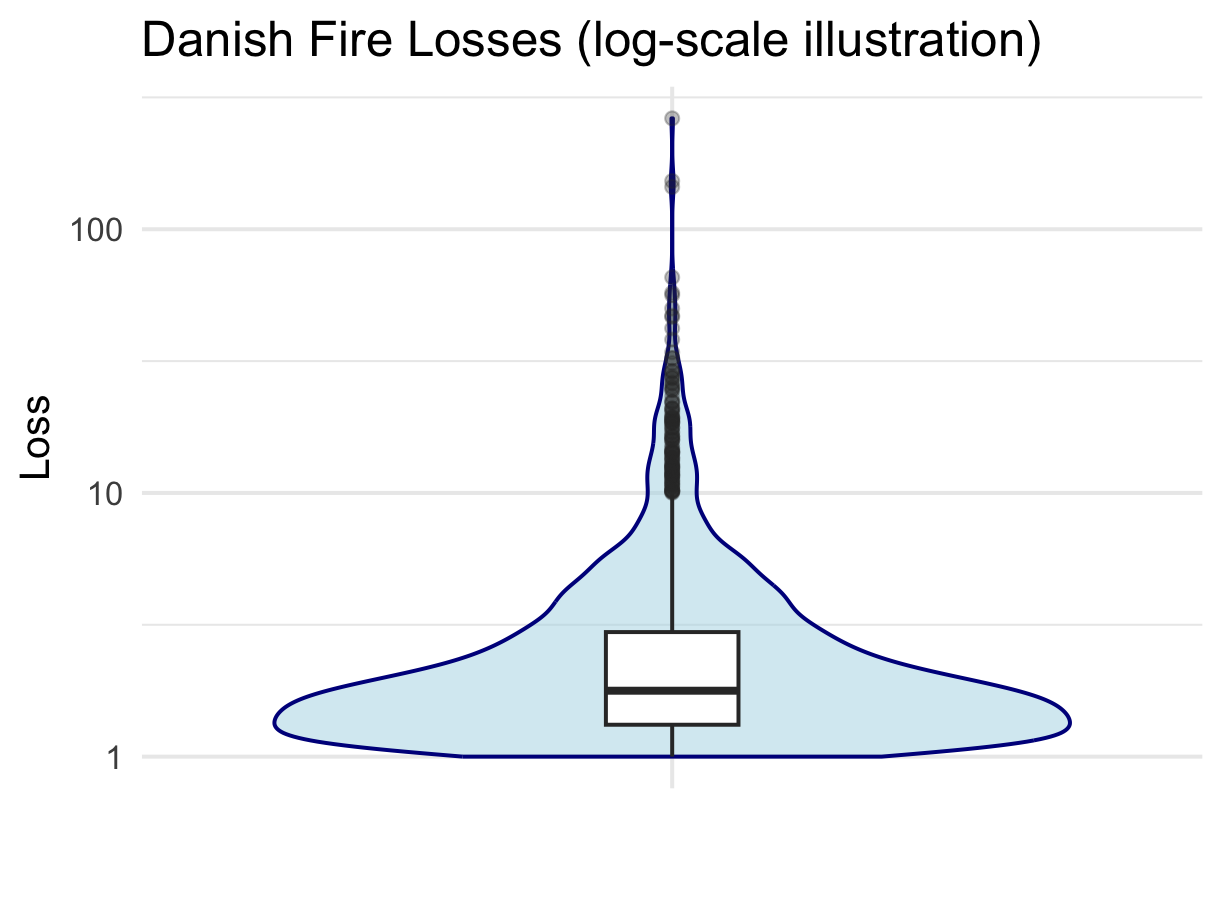}
 	\caption{Histogram with smoothed density curve (left) and violin plot combined with boxplot (right) for the log-scaled Danish fire data.}\label{fig-4}
 \end{figure}	
 
 \begin{figure}[ht!]
 	\centering	
 	\includegraphics[width=2in]{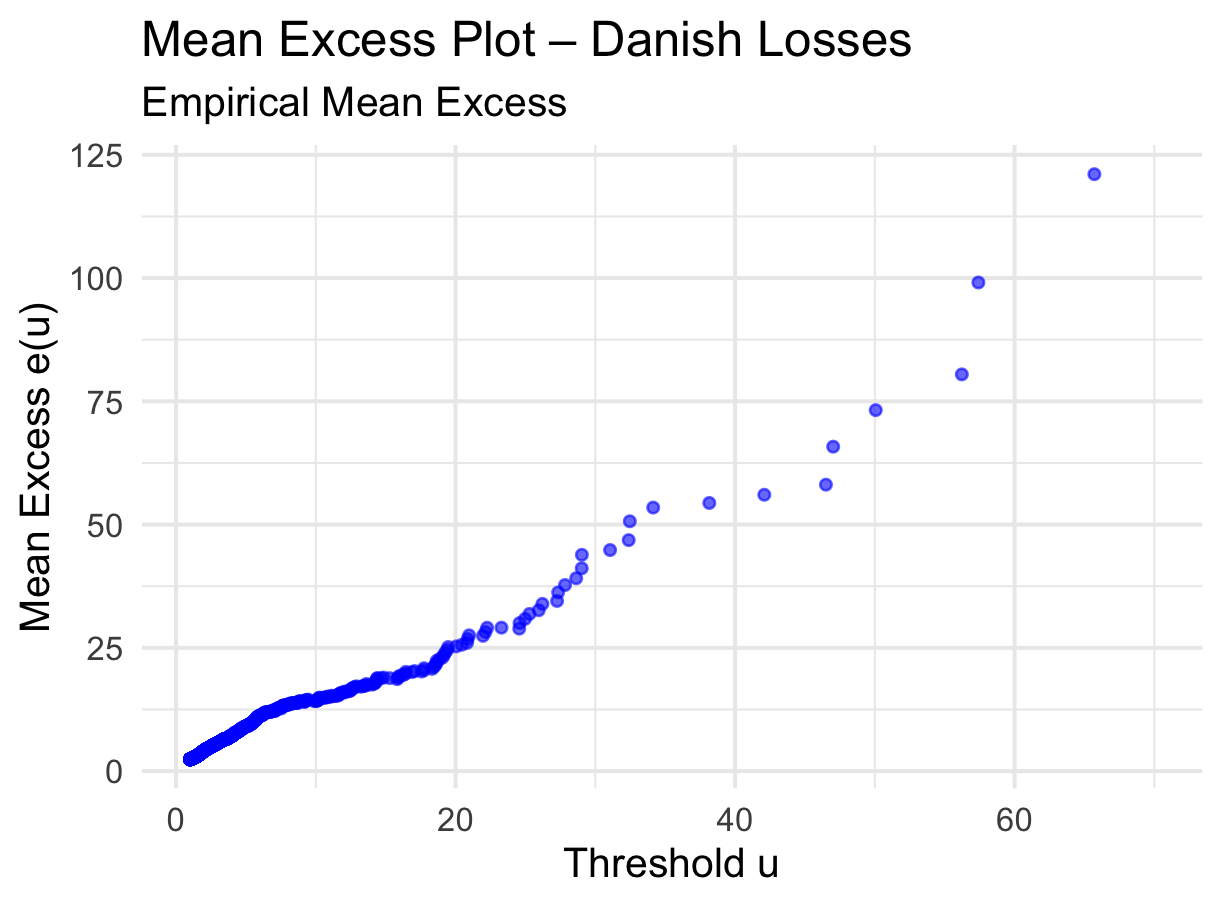}
 	\includegraphics[width=2in]{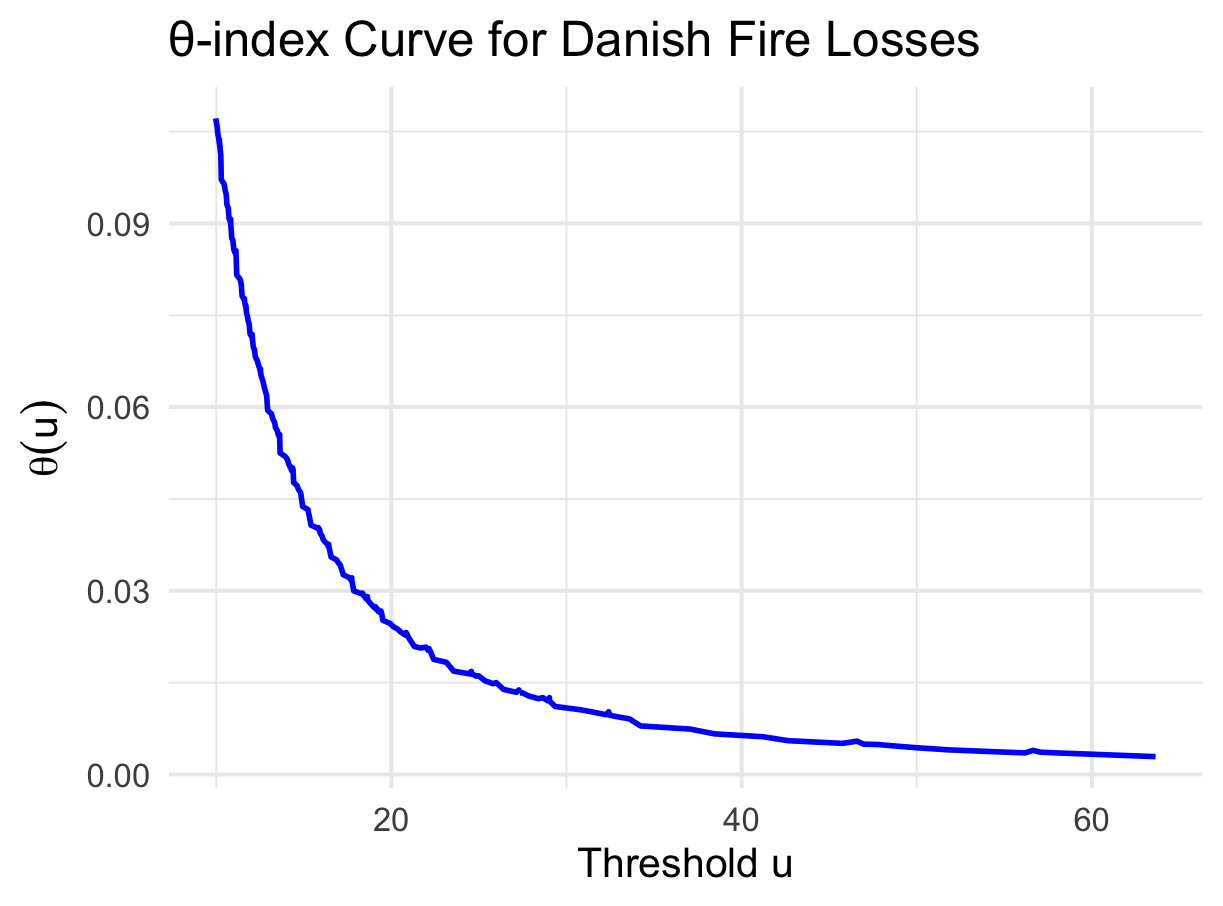}
 	\includegraphics[width=2in]{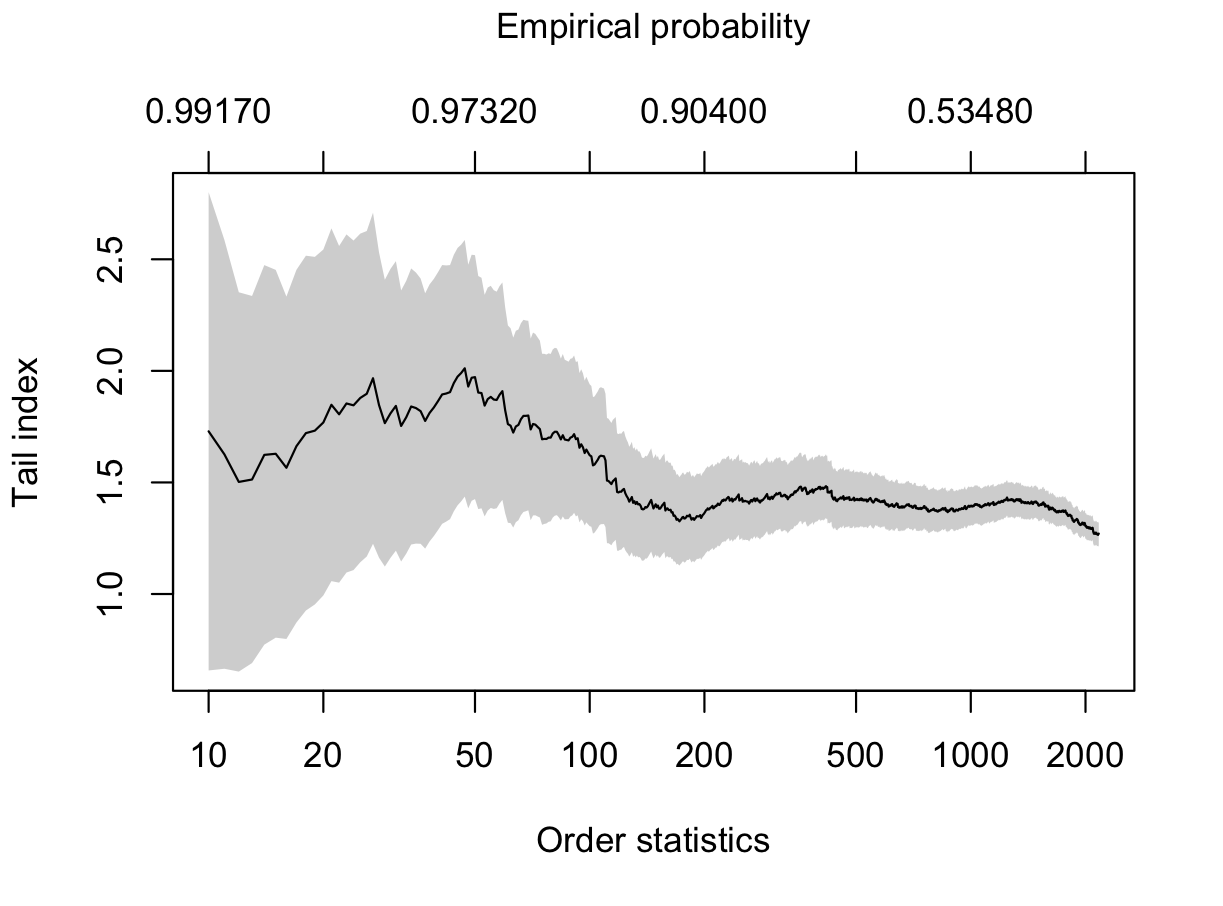}
 	\caption{Mean excess plot (left), $\theta$-index plot (middle) and Hill plot (right) for the Danish fire data.}\label{fig-5}
 \end{figure}
 
 Figure \ref{fig-4} provides a descriptive overview of the Danish fire losses on a logarithmic scale. The distribution is strongly right-skewed, with a high concentration of moderate losses and a small number of very large claims, a pattern typical of heavy-tailed insurance data. The violin and boxplot representation highlights this asymmetry and the presence of extreme observations well separated from the main body of the distribution. Tail behaviour is examined more closely in Figure \ref{fig-5}. The mean excess plot displays a clear increasing trend beyond moderate thresholds, supporting heavy-tailed behaviour. The $\theta$-index curve decreases smoothly as the threshold increases, indicating a gradual change in local tail characteristics rather than a sharp transition at a single level. The Hill plot exhibits a relatively stable region over a range of upper order statistics, suggesting a plausible window for Pareto-type tail approximation. Taken together, these diagnostics motivate a local, threshold-dependent analysis of tail behaviour, in which the $\theta$-index, its decay rate and its curvature provide additional insights into how tail characteristics evolve across different loss levels.

\begin{figure}[ht!]
	\centering
	\includegraphics[width=3in]{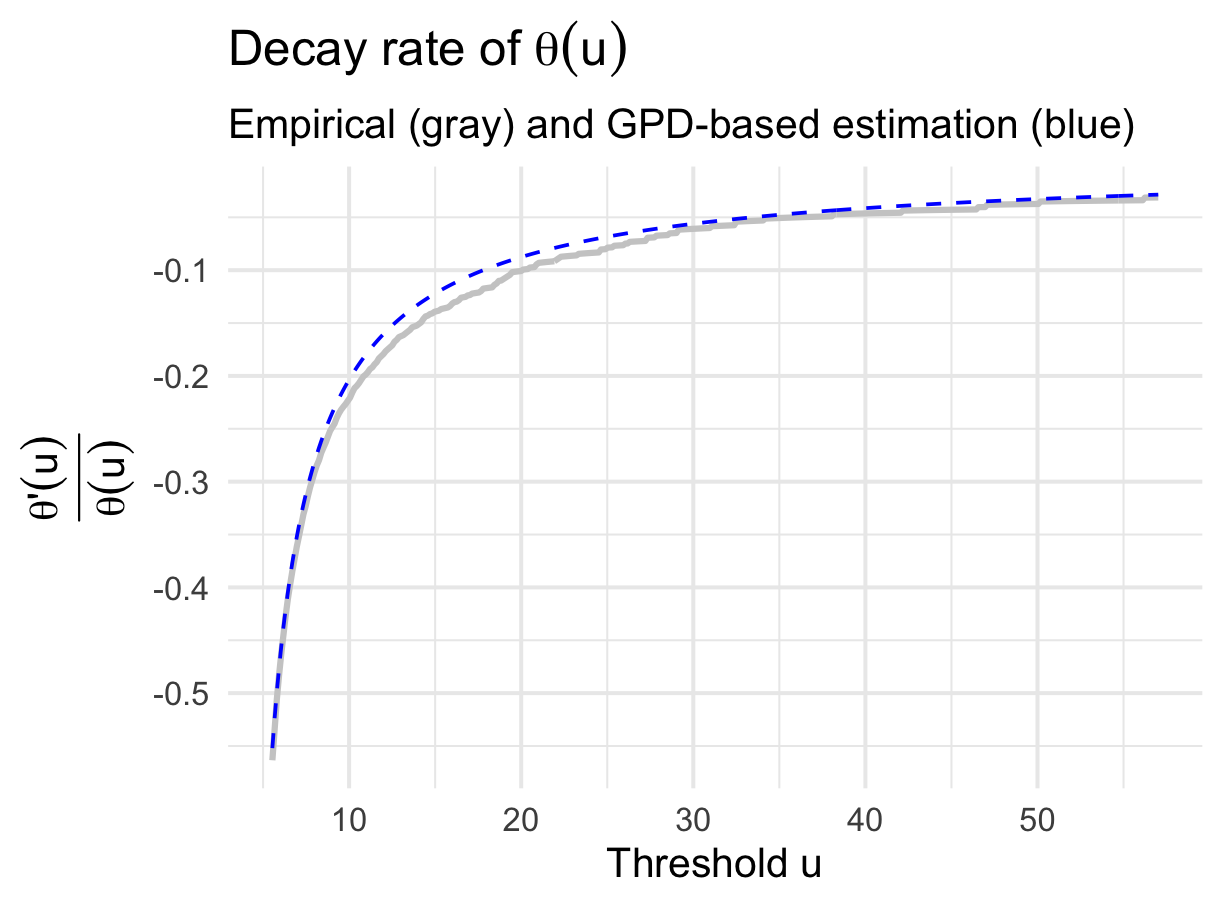}
	\includegraphics[width=3in]{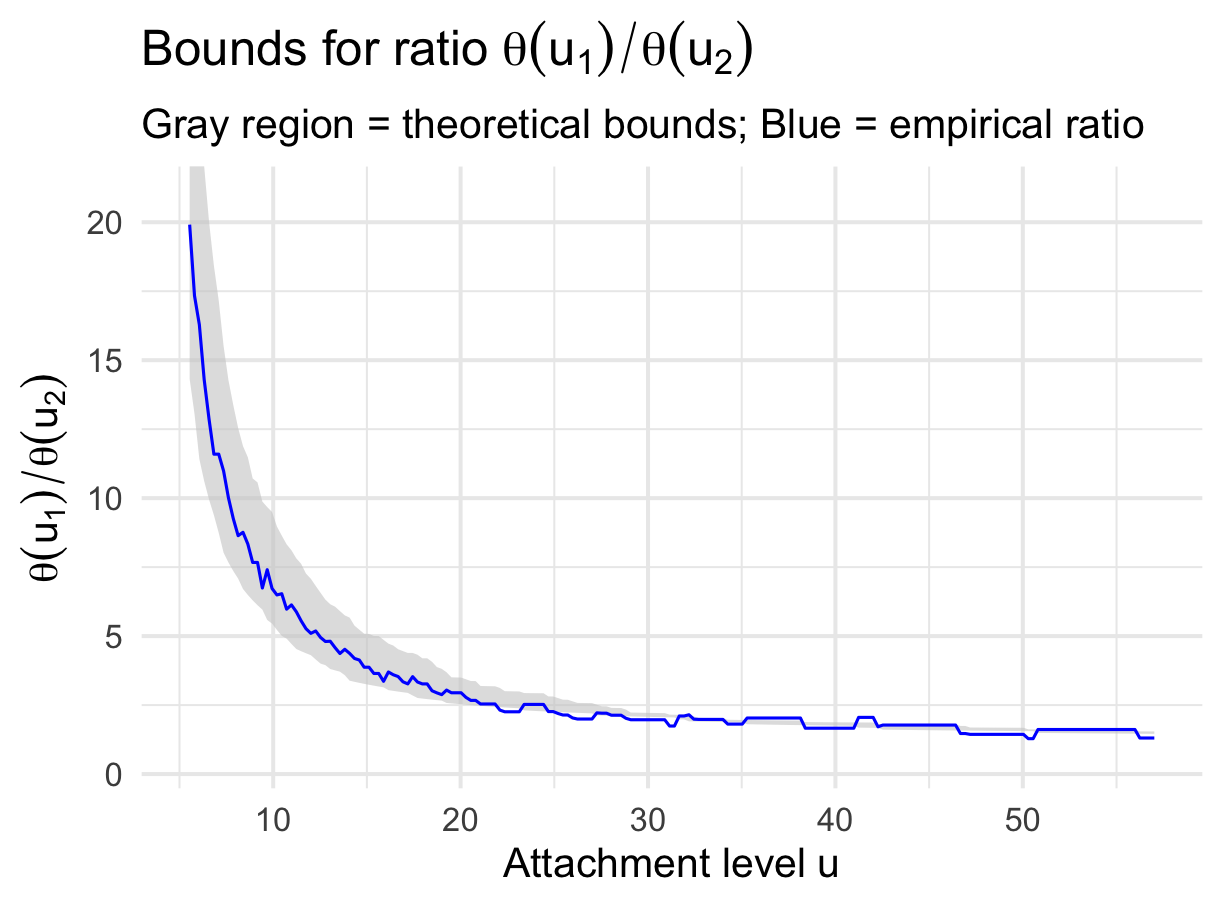}
	\caption{The decay rate of $\theta$-index (left) and the $\theta$-ratio with the theoretical bounds with respect to the attachment levels (right).}\label{fig-6}
\end{figure}

\begin{figure}[ht!]
	\centering	
	\includegraphics[width=3in]{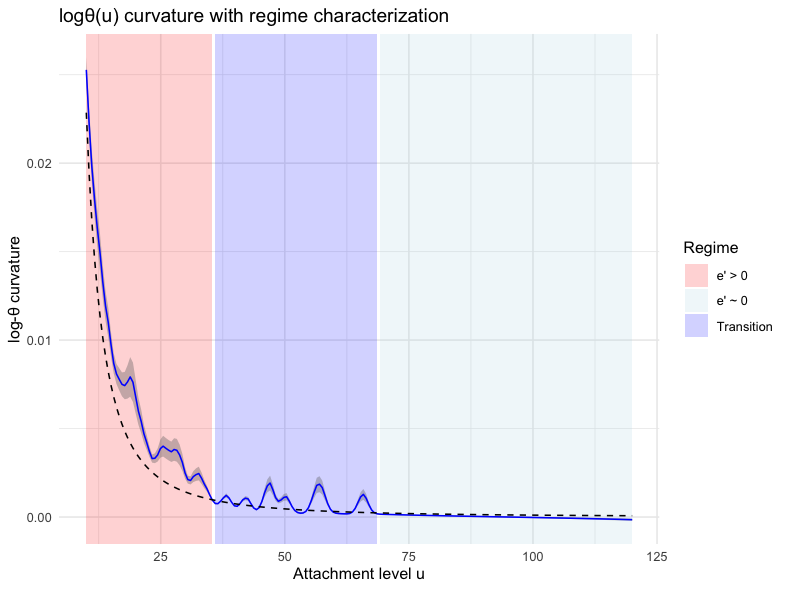}
	\includegraphics[width=3in]{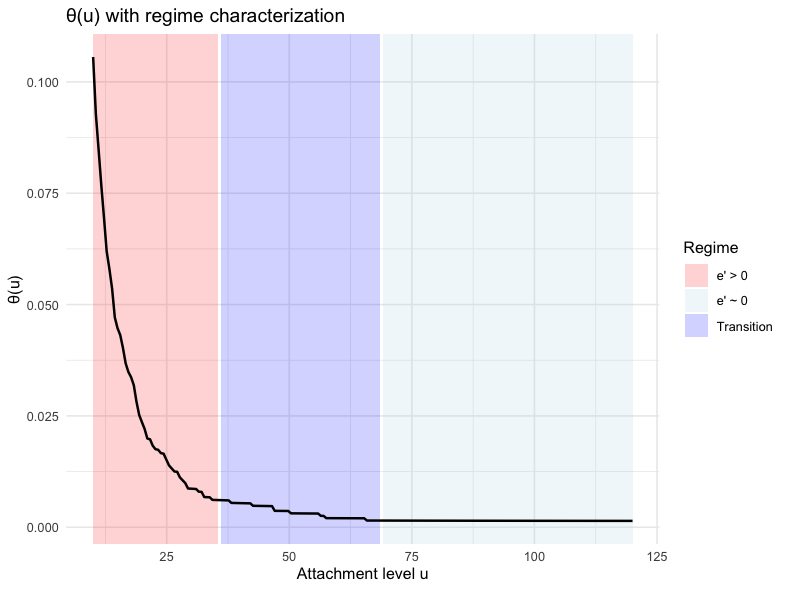}
	\includegraphics[width=3in]{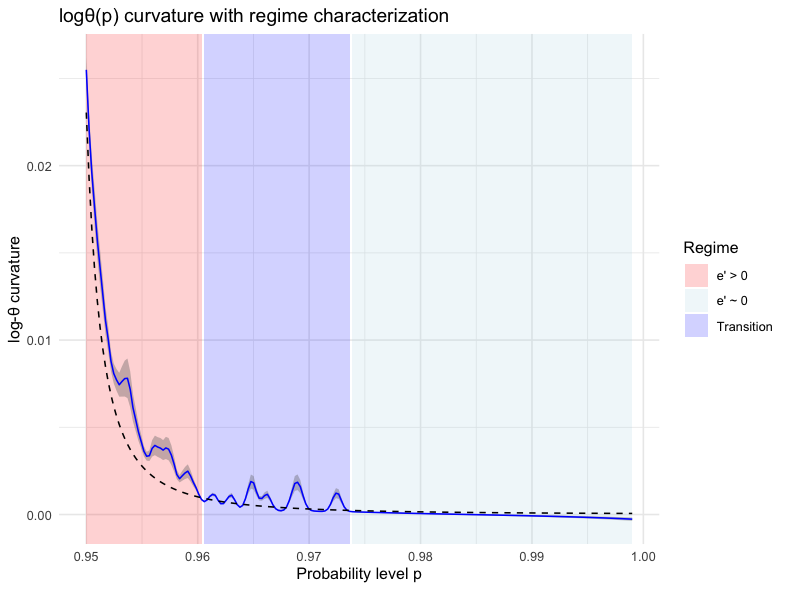}
	\includegraphics[width=3in]{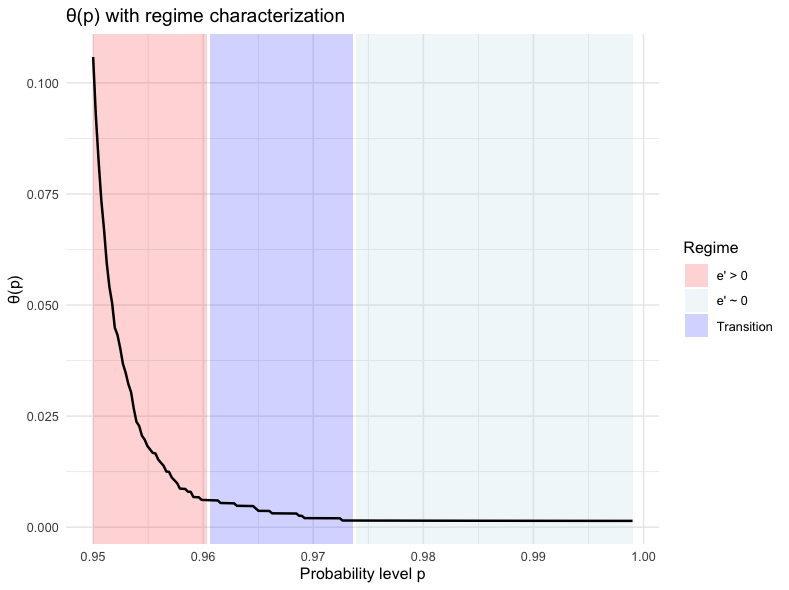}
	\caption{Illustration of the curvature of $\log \theta$ (dashed-lines correspond to the GPD-based estimates) and the $\theta$-index with tail regime distinction, shown both as functions of the threshold $u$ (upper panel) and of the corresponding probability level $p$ (lower panel). The two representations provide complementary views in terms of attachment levels and tail probabilities.}\label{fig-7}
\end{figure}

Figure \ref{fig-6} reports the empirical behaviour of the $\theta$-index and related quantities across increasing threshold levels, together with corresponding GPD-based estimates obtained from a generalized Pareto distribution fitted above a fixed high threshold (the 90\% empirical quantile). The left plot shows the decay rate of $\theta(u)$, which decreases smoothly as the threshold increases and is well approximated by the fitted GPD over a broad range of $u$, indicating a gradual stabilization of tail behaviour. The right plot compares the empirical ratio $\theta(u_1)/\theta(u_2)$ with the theoretical bounds derived in Proposition \ref{prop-10}. While small deviations from the bounds are observed at higher attachment levels, reflecting increased sampling variability, the empirical ratio remains largely consistent with the theoretical envelope over the range where a sufficient number of exceedances is available. Taken together, these plots summarise how tail behaviour varies with the threshold and place the $\theta$-based quantities in direct relation to classical EVT diagnostics.

Figure \ref{fig-7} shows how the curvature of $\log \theta$ evolves across increasing thresholds and how this behaviour is reflected in the corresponding values of $\theta$-index. At lower threshold values the curvature is clearly separated from the constant boundary ($e'_X(u)=0$), while over an intermediate range it fluctuates around this boundary before remaining close to it at higher levels. This progressive flattening of the curvature is consistent with the transition toward a near-constant regime, in the sense of Remark \ref{rmk-5}. The corresponding $\theta$-index curve, display a consistent change in slope across these ranges, with a marked decrease at lower threshold values and a progressive stabilization thereafter. The same pattern is observed when the analysis is expressed in terms of probability levels, and the most pronounced changes occur for probability levels above approximately $p \simeq 0.95$, confirming that the analysis effectively focuses on the tail of the loss distribution.

\section{Conclusions}\label{Sec-8}

This paper introduced the $\theta$-index as a level-dependent measure of tail shape derived from an equal level relationship between VaR and FES. The index provides a scale-free description of upper tail behaviour with a direct interpretation through the mean excess function. It induces a partial order on loss distributions, with comparisons characterized through the monotonicity condition of associated right tail spread ratios. By construction, the $\theta$-index parametrizes a probability equal level representation of VaR as a mixture of ES and the mean. This representation is subsequently used to derive explicit Euler allocations for aggregate risks. The threshold-based formulation further supports the analysis of tail behaviour across attachment levels. An empirical application to the Danish fire data illustrates how the $\theta$-index complements standard tail diagnostics in actuarial practice.

\subsubsection*{Author contributions} \emph{Both authors contributed equally to this manuscript.} 

%% ==================================================================

\bibliographystyle{chicago}
\bibliography{references}
		
\end{document}